\documentclass[11pt]{article}
\usepackage{fullpage}
\usepackage[utf8]{inputenc}
\usepackage{algorithm}
\usepackage{algorithmic}
\usepackage{amsfonts}
\usepackage{graphicx}
\usepackage{amsthm}
\usepackage{mathrsfs}
\usepackage{amssymb}
\usepackage{hyperref}
\usepackage{amsmath}
\usepackage{enumitem}
\usepackage{xcolor}
\usepackage{multirow}
\usepackage{subfigure}
\usepackage{cleveref}
\usepackage{array}
\usepackage{mathtools}

\usepackage[numbers]{natbib}

\usepackage{authblk}
\usepackage{booktabs}
\usepackage{pbox}

\usepackage{xspace} % for \xspace
\usepackage{nicefrac} % for \nicefrac

\newtheorem{theorem}{Theorem}
\newtheorem{lemma}{Lemma}

\newtheorem{corollary}{Corollary}

\newtheorem{example}{Example}
\newcommand{\set}[1]{\left\{ #1 \right\}}
\newcommand{\rpar}[1]{\left( #1 \right)}

\newcommand{\abs}[1]{\left| #1 \right|}

\def\e#1{\emph{#1}}
\def\A{\mathbf{A}}
\def\P{\mathbf{P}}
\def\S{\mathcal{S}}

\def\D{\Delta} % distance between scores of two committees (formerly \text{D})
\def\MD{\overline{\Delta}} % max distance to any other committee (formerly \textMD})
\def\Ck{\mathcal{C}_k}
\def\comp{\mathcal{C}} % set of all completions

\def\top{\text{Top}}
\def\middle{\text{Middle}}
\def\bottom{\text{Bottom}}

\def\neccom{\mathsf{NecCom}}
\def\poscom{\mathsf{PosCom}}
\def\necmem{\mathsf{NecMem}}
\def\posmem{\mathsf{PosMem}}
\def\necjr{\mathsf{NecJR}}
\def\posjr{\mathsf{PosJR}}
\def\necpjr{\mathsf{NecPJR}}
\def\pospjr{\mathsf{PosPJR}}
\def\necejr{\mathsf{NecEJR}}
\def\posejr{\mathsf{PosEJR}}
\def\necpjrp{\mathsf{NecPJR+}}
\def\pospjrp{\mathsf{PosPJR+}}
\def\necejrp{\mathsf{NecEJR+}}
\def\posejrp{\mathsf{PosEJR+}}
\def\angs#1{\mathord{\langle{#1\rangle}}}
\def\neccompar#1{\neccom\angs{#1}}
\def\poscompar#1{\poscom\angs{#1}}
\def\necmempar#1{\necmem\angs{#1}}
\def\posmempar#1{\posmem\angs{#1}}
\def\necpjrpar#1{\necpjr\angs{#1}}
\def\pospjrpar#1{\pospjr\angs{#1}}
\def\necejrpar#1{\necejr\angs{#1}}
\def\posejrpar#1{\posejr\angs{#1}}

\newcommand{\rrule}{\ensuremath{r}} % a voting rule
\newcommand{\tsc}{\mathrm{sc}} % Thiele-score

\newcommand{\naturals}{{{\mathbb{N}}}}
\newcommand{\rationals}{{{\mathbb{Q}}}}
\newcommand{\reals}{{{\mathbb{R}}}}

\newcommand{\abc}{ABC\xspace}

\newcommand{\thiele}{Thiele\xspace}
\newcommand{\pav}{PAV\xspace}

\newcommand{\cc}{CC\xspace}

\newcommand{\av}{AV\xspace}

\newcommand{\ie}{i.e.,\ }
\newcommand{\eg}{e.g.,\ }

\newcommand{\wilog}{w.l.o.g.\ }

\newcommand{\wrt}{w.r.t.\ }

\newcommand{\eat}[1]{}

\def\aviram#1{{\color{green}\textbf{AI:} #1}}
\def\jonas#1{{\color{teal}\textbf{JI:} #1}}
\def\benny#1{{\color{blue}BK: #1}}

\def\aviram#1{}
\def\jonas#1{}
\def\benny#1{}
\def\mb#1{}

\begin{document}

\title{Approval-Based Committee Voting under Incomplete Information}

\author[1]{\large Aviram Imber}
\author[2]{\large Jonas Israel}
\author[2,3]{\large Markus Brill}
\author[1]{\large Benny Kimelfeld}

\affil[1]{\normalsize Technion -- Israel Institute of Technology, Haifa, Israel}
\affil[2]{\normalsize Research Group Efficient Algorithms, TU Berlin, Germany}
\affil[3]{\normalsize Department of Computer Science, University of Warwick, UK}

\sloppy

\date{}

\maketitle

\begin{abstract}
  We investigate approval-based committee voting with incomplete information about the approval preferences of voters. We consider several models of incompleteness where each voter partitions the set of candidates into \e{approved},
  \e{disapproved}, and \e{unknown} candidates, possibly
  with ordinal preference constraints among candidates in the latter category. 
  This captures scenarios where voters have not evaluated all candidates and/or it is unknown where voters draw the threshold between approved and disapproved candidates. 
  We study the complexity of some fundamental computational problems for a number of classic approval-based committee voting rules including Proportional Approval Voting and Chamberlin--Courant.
  These problems include determining whether a given set of candidates is a possible or necessary winning committee and whether a given candidate is possibly or necessarily a member of the winning committee.
  We also consider proportional representation axioms and the problem of deciding whether a given committee is possibly or necessarily representative. 
\end{abstract}

\section{Introduction}
Approval-based committee (ABC) voting represents a well-studied multiwinner election setting, where a subset of candidates of a predetermined size, a so-called \e{committee}, needs to be chosen based on the approval preferences of a set of voters \citep{LaSk20a}.
Traditionally, ABC voting is studied in the context where we know, for each voter and each candidate, whether the voter approves the candidate or not. 
In this paper, we investigate the situation where the approval information is incomplete. Specifically, we assume that each voter is associated with a set of \e{approved} candidates, a set of \e{disapproved} candidates, and a set of candidates where the voter's stand is unknown, hereafter referred to as the \e{unknown} candidates. Moreover, we may have (partial) ordinal information on voters' preferences among the unknown candidates, restricting the ``valid'' completions of voters' approval sets.

When the number of candidates is large, unknown candidates are likely to exist because voters are not aware of or not familiar with, and therefore cannot evaluate, all candidates. 
In particular, this holds in scenarios where candidates join the election over time, and voter preferences over new candidates have not been elicited \citep{CLM+12a}. Distinguishing between disapproved and unknown candidates accounts for two fundamentally different reasons for non-approvals: either the voter has evaluated the candidate and judged him or her not worth approving (in which case the candidate counts as disapproved), or the voter has not evaluated the candidate (in which case the candidate counts as unknown).
Furthermore, incorporating (partial) ordinal preferences among unknown candidates allows us to model situations in which we (partially) know how a voter rank-orders the candidates, but we do not know where the voter draws his or her ``approval threshold.''   

Scenarios involving incomplete knowledge about approval preferences arise naturally in a variety of practical settings. 
For example, in a scenario where we retrieve information from indirect sources such as social media (say, for the sake of prediction), we may get only sparse information about approval (``Vote for $X$'') and disapproval
(``Definitely not Y''), and possibly pairwise preferences due to explicit statements (e.g., ``$X$ is at least better than $Y$''). 
As another example, in shortlisting scenarios (such as faculty hiring) some voters may know with sufficient confidence that they support or oppose some candidates (e.g., the ones from their own field of expertise) but have no clear opinion on others. 
This situation also naturally arises when labeling documents for information retrieval, where the committee corresponds to the page of search results: some documents are clearly relevant, some clearly irrelevant, and some are unclear. For the unclear ones, it may be way easier to rank the documents (totally or partially)  by relevance rather than to insist on a complete classification into relevant and irrelevant; in fact, this is the motivation behind some successful methodologies for learning ranking functions (\e{learning to rank}) such as the pairwise and listwise approaches~\citep{DBLP:conf/icml/CaoQLTL07}.

Our basic model of incompleteness is the \e{poset approval} model that is illustrated in Figure~\ref{fig:model_incomplete}(a).  This model is a rather direct generalization of the voter model in the seminal work of  \citet{KoLa05a}, who study problems of winner determination with incomplete preferences under ranking-based single-winner rules (such as plurality and Borda). 
For each voter, we are given a set of \e{approved} candidates and a set of \e{disapproved} candidates, together with a partial order over the remaining (``unknown'') candidates that constrains the possible approval ballots of the voter: 
if an unknown candidate is preferred to another unknown candidate, then the former needs to be approved by the voter whenever the latter is approved.
In other words, each possible world is obtained by selecting a linear extension of the partial order and determining a cutoff point---every candidate before the cutoff is approved, and every candidate after the cutoff is disapproved (in addition to the known approvals and disapprovals, respectively).  We study in depth two special cases of this model that correspond to the two extremes of posets: zero information (3VA) and full information (linear).
\begin{figure}[t]
    \centering
    \subfigure[Poset approval]{\label{fig:poset} \includegraphics[width=0.4\textwidth]{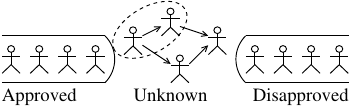}  }\\
    \subfigure[Three-valued approval (3VA)]{\label{fig:3vl} \includegraphics[width=0.4\textwidth]{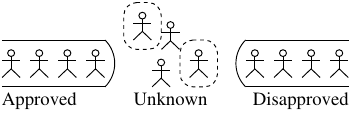} }
    \hskip4em
    \subfigure[Linear incomplete approval]{\label{fig:linear} \includegraphics[width=0.4\textwidth]{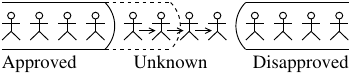} }
    \caption{Models of incomplete approval profiles: Poset approval and the special cases of 3VA and linear incomplete approval. Dashed frames depict candidates that are approved in a valid completion.}
    \label{fig:model_incomplete}
\end{figure}
\begin{itemize}
    \item In the \e{three-valued approval} (3VA) model,\footnote{The term is analogous to ``Three-Valued Logic'' (3VL) that is adapted, e.g., in SQL~\citep{DBLP:journals/tods/Libkin16}, where the three values are \e{true}, \e{false} and \e{unknown}.} illustrated in Figure~\ref{fig:model_incomplete}(b), the partial order of the poset is empty. Hence, every subset of unknown candidates determines a valid possible world where this set is approved and its complement is disapproved.
    \item In the \e{linear} model, illustrated in Figure~\ref{fig:model_incomplete}(c), the unknown candidates are ordered linearly. Hence, every prefix of this order determines a possible world where the candidates in this prefix are approved and the ones in the remaining suffix are disapproved. 
\end{itemize}
From the computational perspective, the models are fundamentally different. For instance, a voter can have an exponential number of possible worlds (or \emph{completions}) in one case, but only a linear number in the other. 
Of course, when a problem is tractable in the poset approval model, then it is also tractable in the 3VA and linear models. On the contrapositive,  whenever a problem is intractable (e.g., NP-hard) in one of these two models, it is also intractable for the general poset approval model. However, as we illustrate later, a problem may be tractable in 3VA and intractable in the linear model, and vice versa.

We investigate the computational complexity of fundamental problems that arise in ABC voting with incomplete approval profiles. In the first problem, the goal is to determine whether a given set $W$ of candidates is a \e{possible} winning committee. In the second problem, the goal is to determine whether such given $W$ is a \e{necessary} winning committee. More formally, a \e{possible committee} (respectively, \e{necessary committee}) is a set of candidates that is a winning committee in \e{some} completion (respectively, \e{all} completions) of the incomplete approval profile. 
We also achieve some results on the problem where we are given a candidate and the goal is to determine whether the candidate is possibly or necessarily a member of a winning committee. 

In the applications described earlier it is often enough to find only one winning committee, e.g., because we only want to have one set of top search results and do not care about there being other, equally good such sets. Thus finding a necessary winning committee is enough and we do not need to elicit any more information. Conversely, if for example in a hiring process we know that some candidates are not in any possible committee any more, we can already inform them that they will definitely not get the position and we can stop eliciting preferences over those candidates.

We consider a class of voting rules that was introduced by 
\citet{Thie95a}. This broad class of approval-based committee voting rules contains several classic rules such as Approval Voting (AV),  Chamberlin--Courant (CC), and Proportional Approval Voting (PAV).
For all Thiele rules except AV, the problem of determining winning committees is intractable even before we allow for incompleteness~\citep{LaSk20a}. Consequently, computing either possible or necessary committees is intractable as well, under each of our models of incompleteness. Therefore, we focus on the case where the committee size is small (i.e., constant). In this case, winning committees can be computed in polynomial time for all considered rules.

Finally, we study questions about proportional representation of voters given incomplete approvals. We consider the property of Justified Representation (JR), which requires that all voter sets that are large enough and cohesive enough are represented in some way in the committee~\cite{ABC+16a}.
We investigate the complexity of deciding, given an incomplete approval profile, whether it is possible or necessary for a set of candidates to form a committee that satisfies JR. By ``possible'' we mean that JR is satisfied in at least one completion, whereas ``necessary'' means that JR is satisfied in every completion.  We show that both possibility and necessity problems are solvable in polynomial time, even if the committee size is not assumed to be fixed. 
We also consider the two stronger notions of Proportional Justified Representation (PJR)~\cite{SFF+17a} and Extended Justified Representation (EJR)~\cite{ABC+16a}, and their more recent stronger versions PJR+ and EJR+~\cite{BrPe23a}. 
For PJR+ and EJR+, we show that the necessity problems are solvable in polynomial time, again without assuming that the committee size is fixed; this is also the case for the possibility problems under the 3VA model, while the cases of the other two models remain open.
We also present polynomial-time algorithms for the necessity problems of EJR and PJR, assuming that the committee size is fixed;
the complexity of the possibility problems remain open in all three models of incompleteness.

\paragraph{Overview.}
The remainder of this paper is structured as follows. We discuss related work in \Cref{sec:abcu_related} and define relevant notation and voting rules in \Cref{sec:abcu_prelims}. We formally introduce the models of uncertainty and the computational problems we consider in \Cref{sec:models}. \Cref{sec:poscom,sec:neccom} contain our results for computing possible and necessary winning committees, respectively. Finally, we consider proportional representation axioms in \Cref{sec:representation}, before \Cref{sec:conclusion} concludes. 

\section{Related Work}\label{sec:abcu_related}

There are at least three streams of literature to which our work is related. 

\paragraph{Possible and necessary winners.}
The poset model, one of our three models of incompleteness described above, is a rather direct generalization of the voter model in the work of  \citet{KoLa05a}, who study problems of winner determination with incomplete preferences under ranking-based singlewinner rules (such as plurality and Borda) and introduced the concepts of possible and necessary winners. 
A recent overview of the literature on possible and necessary winner problems as well as related concepts is given by \citet{Lang20a}. We only consider works that are particularly close to our approaches here.
Following the paper by \citet{KoLa05a}, a complete dichotomy of the complexity of the possible and necessary winner problems for partial orders has been established in a series of papers~\citep{DBLP:journals/ipl/BaumeisterR12,DBLP:journals/jcss/BetzlerD10,DBLP:journals/jair/XiaC11}. In these papers it is shown that under every positional scoring rule in the setting of partial orders, the necessary winners can be found in polynomial time, yet it is NP-complete to decide whether a candidate is a possible winner (assuming a regularity condition that the rule is pure), except for the tractable cases of the plurality and veto rules. 
\citet{CLM+12a} study another variant of the possible and necessary winner problem \wrt scoring rules. They consider a setting with (complete) ordinal ballots where new candidates are added and the question is whether the initial candidates are possible or necessary winners. 
Possible and necessary winners under \emph{truncated ballots} are also somewhat related to our model of incomplete approval preferences. Under truncated ballots, voters only give a ranking over the most preferred candidates instead of ranking all candidates. Possible and necessary winner problems given these kinds of partial ballots are studied by \citet{BFLR12a}, \citet{KKKG11a}, and \citet{ayadi2019single}.

\paragraph{Approval voting with incomplete information.}
Particularly close to our setting is the work by \citet{BGL+13a}. Here, possible and necessary winner problems for singlewinner \av have been studied. They use a preference model in which every voter is assumed to have a total ranking of the candidates, but it is unknown where each voter positions their ``approval threshold''. The authors also extend some of their results to the multiwinner setting. These results overlap with some of our results for the linear model: \citet{BGL+13a} only consider the linear model and show that the possible winner problem is NP-complete for AV. We discuss specific relationships between their work and ours in Section~\ref{sec:poscom}. 
\citet{benabbou2016solving} also consider possible and necessary winner problems under approval preferences. The main difference to our work is that in this work the authors consider knapsack voting and assume preferences by voters over different knapsacks (which they elicit incrementally), while we consider preferences over candidates (\ie over single items in the case of knapsack voting).
\citet{terzopoulou2021restricted} consider restricted domains for approval voting under uncertainty using a model of incompleteness corresponding to our 3VA model.
Another approach to \abc voting with incomplete preference information is the recent paper by \citet{HKP+23a}. In their setting, individual voters are queried on their approval preferences over a subset of the candidates and the goal is to construct query procedures that find representative committees (\wrt to the unknown complete approval profile).

\paragraph{Trichotomous preferences.}
Allowing voters to not only specify which candidates they approve and disapprove, but also for which they are undecided, can be seen as eliciting trichotomous preferences (without uncertainty). Under trichotomous preferences, voters note for each candidate whether they approve the candidate, are indifferent about it, or disapprove of the candidate. So while dichotomous preferences are equivalent to a partition of the candidates into two set, trichotmous preferences correspond to a partition into three sets. \citet{Fels89a} compared \av over dichotomous preferences with its natural extension\footnote{For \av under trichotomous preferences, we count each approval by a voter as one point for a candidate but deduct one point for each disapproval. The candidate(s) with the most aggregated points are the winners.} to trichotomous preferences. \citet{AlLa14a} characterize \av on trichotomous preferences. The minimax-\av rule under trichotomous preferences (among other preferences) is studied by \citet{BaDe15b}. \citet{ZYG19a} investigate the computational complexity of winner determination for \abc rules such as \pav and \cc under trichotomous preferences. Finally, proportionality axioms for trichotomous preferences (including adapted versions of JR and PJR) are studied by \citet{talmon2021proportionality}.

\section{Preliminaries}
\label{sec:abcu_prelims}
We consider a voting setting where a finite set $V = \set{1, \ldots, n}$ of \emph{voters} has preferences over a finite set $C$ of \emph{candidates}. We usually use $n$ for the number of voters while the number of candidates is denoted by $m$.
An \emph{approval profile} $\A = (A(1), \dots, A(n))$ lists the approval sets of the voters, where $A(i) \subseteq C$ denotes the set of candidates that are approved by voter $i$. The concatenation of two approval profiles $\A_1=(A_1(1), \ldots, A_1(p))$ and $\A_2=(A_2(1), \ldots, A_2(q))$ is denoted by $\A_1 \circ \A_2 \coloneqq (A_1(1), \ldots, A_1(p), A_2(1), \ldots, A_2(q))$. An \e{approval-based committee (ABC) rule} takes an approval profile as input and outputs one or more size-$k$ subsets of candidates, so called \emph{committees}, for a given parameter $k \in \naturals$.
We refer to the committee(s) output by a rule $r$ as the \textit{winning committee(s)} (w.r.t.\ the rule $r$ and committee size $k$). 

Most of our work is focused on a class of ABC rules known as \e{Thiele rules} \citep{Thie95a,BLS18a}. These rules are parameterized by a \textit{weight function} $w$, i.e., a non-decreasing function $w \colon \naturals \rightarrow \rationals_{\ge 0}$ with $w(0) = 0$. The score that a voter~$v$ contributes to a subset $S \subseteq C$ is defined as $w(|S \cap A(v)|)$ and the score of $S$ is $\sum_{v \in V} w(|S \cap A(v)|)$. The rule \textit{$w$-Thiele}  then outputs the subset(s) of size $k$ with the highest score. Two of the most famous examples of Thiele rules are \emph{Approval Voting (AV)},\footnote{We can also define this rule using scores of individual candidates: For a candidate $c$, the \textit{approval score} of $c$ is defined as $|\set{v \in V : c \in A(v)}|$. AV selects committees consisting of the $k$ candidates with highest approval scores.} where $w(x) = x$, and \textit{Proportional Approval Voting (PAV)}, where $w(x) = \sum_{i=1}^x 1/i$. We assume that $w(x)$ is computable in polynomial time in $x$, and that there exists an integer $x$ with $w(x) \neq 0$ (otherwise, all subsets of voters have the same score). Observe that for every \e{fixed} number $k$, we can find the committee(s) under $w$-Thiele in polynomial time by computing the score of every subset $W \subseteq C$ of size~$k$.

Without loss of generality, we assume that weight functions are normalized such that $w(x)=1$ for the smallest $x$ with $w(x)>0$. 
We call a Thiele rule \textit{binary} if $w(x)\in\{0,1\}$ for all $x \in \naturals$. Since $w$ is non-decreasing, binary Thiele rules are characterized by an integer $t$ such that $w(x) = 0$ for every $x < t$ and $w(x) = 1$ for all $x \geq t$. The most prominent binary Thiele rule is \textit{Chamberlin--Courant (CC)}, with $w(x) = 1$ for all $x \geq 1$.

For ease of notation we introduce the following. Given set of candidates $S \subseteq C$, an approval profile $\A = (A_1, \ldots, A_n)$ and a weight function $w$ of a Thiele rule, we denote the score that voter $v \in V$ contributes by $\tsc(A_v, S) = w(|S \cap A_v|)$, and we denote the total score of the profile by $\tsc(\A, S) = \sum_{v \in V} \tsc(A_v, S)$.

\section{Models of Incompleteness}
\label{sec:models}

In this section, we formally introduce the three models of incomplete information that we investigate
throughout the paper
(see \Cref{fig:model_incomplete} for an illustration).
We first define \emph{poset approval} as a general model, and then two models constituting special cases. 
In all models, a \emph{partial profile} $\P = (P(1), \dots, P(n))$ consists of $n$ \textit{partial votes}. 

\paragraph{Poset approval.}
For each voter $v \in V$, the partial vote $P_v$ consists of a partition $(\top(v), \middle(v), \bottom(v))$ of~$C$ into three sets and a partial order $\succeq_{v}$\footnote{We use the symbol $\succeq$ to refer to a partial order and not to a weak order (as is done regularly with this symbol). Recall that a \textit{partial order} is a reflexive, antisymmetric, and transitive binary relation.} over the candidates in $\middle(v)$. $\top(v)$ represents the candidates that voter $v$ approves in any case, $\bottom(v)$ represents the candidates that $v$ approves in no case.  
The partial order $\succeq_v$ represents approval constraints on the candidates in  $\middle(v)$, whose approval is open. Specifically, if $c \succeq_{v} c'$ and $v$ approves $c'$, then $v$ also approves $c$.
A \textit{completion} of $\P = (P(1), \dots, P(n))$ is an approval profile $\A = (A(1), \dots, A(n))$ where each $A(v)$ ``completes'' $P_v$. Namely, $\top(v) \subseteq A(v)$, $\bottom(i) \cap A(v) = \emptyset$, and for every pair $c, c' \in \middle(v)$ with $c \succeq_{v} c'$, it holds that $c \in A(v)$ whenever $c' \in  A(v)$. 
Equivalently, we can describe $\A$ in the following way: For each voter $v \in V$ select a subset $\middle_\A(i) \subseteq \middle(v)$ which ``respects'' $\succeq_v$ (\ie $c \succeq_v c'$ and $c' \in  \middle_\A(v)$ imply $c \in \middle_\A(v)$) and $v$ approves exactly $A(v) = \top(v) \cup \middle_\A(v)$.

\paragraph{Three-valued approval (3VA).}
This model is a special case of poset approval where for every voter $v \in V$, the partial order $\succeq_v$ over the candidates in $\middle(v)$ is empty (except for the reflexive part). In other words, $v$ might approve any subset of candidates in $\middle(v)$. Hence, a completion of $\P = (P(1), \dots, P(n))$ is a complete approval profile $\A = (A(1), \dots, A(n))$ such that for every voter $v \in V$ we have $\top(v) \subseteq A(v)$ and $\bottom(v) \cap A(v) = \emptyset$.

\paragraph{Linear incomplete approval.}
This model is a special case of poset approval where for every voter $v \in V$, $\succeq_v$ is a complete linear order $c_{i_1}  \succeq_v \dots \succeq_v c_{i_q}$ on the candidates in $\middle(v)$. To construct a completion $\A$, for every voter $v \in V$ we select a $j \leq |\middle(v)|$ such that $\middle_\A(v) = \set{c_{i_1}, \dots, c_{i_j}}$.

\bigskip
In the linear model of incompleteness, every partial vote has at most $m$ completions, as opposed to the previous two models where even a single partial vote can have $2^m$ completions. In all three models, the number of completions for a partial profile can be exponential in the number of voters~$n$.

\begin{table}[t]
    \centering
    \scalebox{1}{
        \begin{tabular}{cccc}
            \toprule
            voter $v$ & $\top(v)$ & $\middle(v)$ & $\bottom(v)$ \\
            \midrule
            1 & $a, b$ & $c$ & -- \\
            2 & $c$ & $a \rightarrow b$ & -- \\
            3 & $a$ & $c$ & $b$ \\
            4 & -- & $c \rightarrow b \rightarrow a$ & -- \\
            5 & $a, c$ & -- & $b$ \\
            \bottomrule
        \end{tabular}
    }
    \caption{Overview of the incomplete approval profile from \Cref{ex:poscom_neccom}. In every completion of the profile, each voter $v$ approves all candidates in $\top(v)$ and some (possibly empty) subset of $\middle(v)$. An arrow from candidate $a$ to candidate $b$ means that if $b$ is approved, so is $a$.}
    \label{tbl:example_poscom_neccom}
\end{table}

\paragraph{Possible and necessary committees.}
In order to define possible and necessary committees, fix an ABC rule $\rrule$.
Given a committee size $k$ and a partial profile $\P$, a set $W \subseteq C$ of $k$ candidates is a \emph{possible committee} if there exists a completion $\A$ of $\P$ where $W$ is a winning committee (under $\rrule$), and a \emph{necessary committee} if $W$ is a winning committee for every completion $\A$ of $\P$. In the following sections, we investigate the computational complexity of deciding whether a given set of candidates is a possible or necessary committee, for different rules $\rrule$.
We close this section with an example highlighting these definitions.

\begin{example}\label{ex:poscom_neccom}
Consider the instance depicted in \Cref{tbl:example_poscom_neccom} using the model of linear incomplete approvals. Here 5 voters have incomplete approval preferences over the three candidates $a,b$ and $c$. 
Under \av and for a committee size of $k = 2$, $\set{a,b}$ is a possible committee since it wins (via a tie with $\set{a,c}$) if voter 2 approves all three candidates and all other voters $i \neq 2$ approve only the candidates in $\top(i)$. Further, $\set{a,c}$ is a necessary committee since it wins in every completion of the profile. Finally, $\set{b,c}$ wins in no completion and thus is not a possible committee.
\end{example}

\section{Computing Possible Committees}\label{sec:poscom}

The first decision problem we study concerns possible committees.
For an ABC rule $\rrule$ and committee size $k$, consider the following decision problem that we denote by $\poscompar{k}$: 
Given a partial profile $\P$ and a subset $W\subseteq C$ of candidates of size $k$, 
decide whether $W$ is a possible committee.
As mentioned earlier, we parameterize the decision problem by $k$ since otherwise winner determination even in the case of no uncertainty is known to be NP-hard for all Thiele rules except \av \citep{SFL16a}.
Observe that $\poscompar{k}$ is in NP under all Thiele rules, for every fixed $k$, in all three models of uncertainty that we study. This is because, given a complete profile as a witness, we can find the winning committees of a fixed size $k$ under Thiele rules in polynomial time.

Our results concerning the complexity of deciding $\poscompar{k}$ under different models of uncertainty are summarized in \Cref{table:poscom}. 
In particular, we show that under the model of poset approval, $\poscompar{k}$ is NP-complete for all Thiele rules. To obtain this result, we study the complexity of $\poscompar{k}$ in the models of three-valued approval and linear incomplete approval in Sections~\ref{sec:poscom-3va} and \ref{sec:poscom-linear}, respectively. Then, we prove our main result in Section~\ref{sec:poscom-proof}. 

We first introduce two useful lemmata that hold in all three models of uncertainty. The first enables us to extend results from one \thiele rule to another. The second one shows how to extend results from one value of the committee size $k$ to all $k' \geq k$.

\newcolumntype{C}{>{\centering\arraybackslash}m{2.7cm}}
\newcolumntype{E}{>{\centering\arraybackslash}m{3.2cm}}
\def\nl{\par\noindent}

\renewcommand{\arraystretch}{1.25}
\begin{table*}[t]
\small
\centering
\begin{tabular}{CCCCE}
\toprule
\textbf{Model} & AV & CC & PAV & {$w$-Thiele} \\
\midrule
Three-Valued Approval & P  [Thm.~\ref{thm:simplePoscomkApproval}] & NP-c  [Thm.~\ref{thm:simplePoscom2CC}] \nl (for all $k \geq 2$) &  ?  &  NP-c for every binary~$w$  [Cor.~\ref{cor:simplePoscomkThiele1Weight}]  \\ \midrule
Linear Incomplete Approval & NP-c  [Cor.~\ref{cor:linearPoscomkThieleIncreasing}]${}^*$
\nl (for all $k\geq 2$)
& P [Thm.~\ref{thm:linearPoscomkThiele1Weight}] 
& NP-c [Cor.~\ref{cor:linearPoscomkThieleIncreasing}]  (for all $k\geq 2$)  &  
P for every binary $w$  [Thm.~\ref{thm:linearPoscomkThiele1Weight}]; NP-c for every strictly increasing $w$ [Cor.~\ref{cor:linearPoscomkThieleIncreasing}]
\\ \midrule
Poset Approval & NP-c [Thm.~\ref{thm:posetPoscomkThiele}] \nl (for all $k\geq 2$) & NP-c [Thm.~\ref{thm:posetPoscomkThiele}] \nl (for all $k\geq 2$) & NP-c [Thm.~\ref{thm:posetPoscomkThiele}] \nl (for all $k\geq 2$)  &  NP-c for every $w$ [Thm.~\ref{thm:posetPoscomkThiele}] \\
\bottomrule
\end{tabular}
\caption{\label{tab:possibleCommittee} Overview of our complexity results for $\poscompar{k}$. ``NP-c'' means that there exists a $k^*$ such that $\poscompar{k}$ is NP-complete for all $k \geq k^*$. In cases where $k^*$ already equals 2 we mention it explicitly. ``P'' means that $\poscompar{k}$ is solvable in polynomial time for all $k$.
The result marked with ${}^*$ also follows from \citet{BGL+13a}. 
}
\label{table:poscom}
\end{table*}

\begin{lemma}
\label{lemma:thieleReduction}
Let $w_1, w_2$ be two weight functions. Assume there are two integers $k, t \geq 0$ and a strictly increasing linear function $g$ such that $w_2(x+t) = g(w_1(x))$ for every $x \in \set{0, 1, \dots, k}$. Then, in each of the three models of uncertainty, there is a polynomial-time reduction from $\poscompar{k}$ under $w_1$-Thiele to $\poscompar{k+t}$ under $w_2$-Thiele.
\end{lemma}
\begin{proof}
Assume an instance of $\poscompar{k}$ under $w_1$-Thiele specified by $C_1, \P_1 = (P_1(1), \dots, P_1(n))$ and $W \subseteq C_1$ of size $k$. Define an instance of $\poscompar{k+t}$ under $w_2$-Thiele where the candidates are $C_2 = C_1 \cup D$ for $D = \set{d_1, \dots, d_t}$. The profile is $\P_2 = (P_2(1), \dots, P_2(n))$ where we have the same set of voters $V$ and for every voter $v \in V$, $P_2(v)$ is the same as $P_1(v)$ except that $v$ also approves the candidates of $D$. Formally, we add the candidates of $D$ to $\top(v)$, and the rest remains unchanged.  
We now prove that $W$ is a possible committee (of size $k$) for $\P_1$ under $w_1$-Thiele if and only if $W \cup D$ is a possible committee (of size $k+t$) for $\P_2$ under $w_2$-Thiele.

Assume that $W \cup D$ is a possible committee (of size $(k+t)$) for $\P_2$, let $\A_2 = (A_2(1), \allowbreak\dots, A_2(n))$ be a completion of $\P_2$ where $W \cup D$ is a winning committee. Consider the completion $\A_1 = (A_1(1), \dots, A_1(n))$ of $\P_1$ where $A_1(v) = A_2(v) \setminus D$ for every voter $v \in V$. We show that $W$ is a winning committee in $\A_1$.
Let $F \subseteq C_1$ of size $k$. Since $W \cup D$ is a winning committee for $\A_2$, we have $\tsc_{w_2}(\A_2, W \cup D) \geq \tsc_{w_2}(\A_2, F \cup D)$. Every voter of $\A_2$ approves all candidates of $D$, hence for every voter $v \in V$ we have $|(F \cup D) \cap A_2(v)| = |F \cap A_1(v)| + t$, and $|F \cap A_1(v)| \in \set{0, \dots, k}$. Our assumption on $w_1$ and $w_2$ can be equivalently written as $w_1(x) = g^{-1}(w_2(x+t))$ for every $x \in \set{0, 1, \dots, k}$. Note that $g^{-1}$ is also a strictly increasing linear function. We can deduce the following.
\begin{align*}
    & \tsc_{w_1}(\A_1, W) - \tsc_{w_1}(\A_1, F) = \sum_{v \in V} w_1(|W \cap A_1(v)| - w_1(|F \cap A_1(v)|) \\
    &= \sum_{v \in V} g^{-1} \rpar{w_2(|W \cap A_1(v)|+t)} - g^{-1} \rpar{w_2(|F \cap A_1(v)|+t)} \\
    &= \sum_{v \in V} g^{-1} \rpar{w_2(|(W \cup D) \cap A_2(v)|)} - g^{-1} \rpar{w_2(|(F \cup D) \cap A_1(v)|)} \\
    &= g^{-1} \rpar{\tsc_{w_2}(\A_2, W \cup D)} - g^{-1} \rpar{\tsc_{w_2}(\A_2, F \cup D)} \geq 0
\end{align*}
Therefore, $W$ is a winning committee of $\A_1$, and $W$ is a possible committee (of size $k$) for $\P_1$.

Conversely, assume that $W$ is a possible committee (of size $k$) for $\P_1$, let $\A_1 = (A_1(1), \dots, A_1(n))$ be a completion of $\P_1$ where $W$ is a winning committee. Define a completion $\A_2 = (A_2(1), \dots, A_2(n))$ of $\P_2$ where $A_2(v) = A_1(v) \cup D$ for every voter $v \in V$. We show that $W \cup D$ is a winning committee of $\A_2$. Let $H \subseteq C_2$ be a subset of size $k+t$. Since $|D| = t$, there are at least $k$ candidates in $H$ which are not in $D$ (these are candidates of $C_1$), let $F = \set{f_1, \dots, f_k} \subseteq H \setminus D \subseteq C$. For every voter $v \in V$, we have $D \subseteq A_2(v)$, \ie $v$ approves all candidates of $D$ in $\A_2$, therefore $|(F \cup D) \cap A_2(v)| \geq |H \cap A_2(v)|$. Since $w_2$ is non-decreasing, we get $\tsc_{w_2}(A_2(v), F \cup D) \geq \tsc_{w_2}(A_2(v), H)$ for every voter $v \in V$, and overall $\tsc_{w_2}(\A_2, F \cup D) \geq \tsc_{w_2}(\A_2, H)$.

Next, recall that $W$ is a winning committee of $\A_1$ (under $w_1$-Thiele), and hence $\tsc_{w_1}(\A_1, W) \geq \tsc_{w_1}(\A_1, F)$. Again, for every voter $v \in V$ we have $|(F \cup D) \cap A_2(v)| = |F \cap A_1(i)| + t$. Therefore
\begin{align*}
    & \tsc_{w_2}(\A_2, W \cup D) - \tsc_{w_2}(\A_2, H) \geq \tsc_{w_2}(\A_2, W \cup D) - \tsc_{w_2}(\A_2, F \cup D) \\
    &= \sum_{v \in V} w_2(|(W \cup D) \cap A_2(v)|)  - w_2(|(F \cup D) \cap A_2(v)|) \\
    &= \sum_{v \in V} g(w_1(|(W \cup D) \cap A_2(v)| - t)) - g(w_1(|(F \cup D) \cap A_2(v)| - t)) \\
    &= \sum_{v \in V} g(w_1(|W \cap A_1(v)| )) - g(w_1(|F \cap A_1(v)|)) \\
    &= g(\tsc_{w_1}(\A_1, W)) - g(\tsc_{w_1}(\A_1, F)) \geq 0 \,.
\end{align*}
We can deduce that $W \cup D$ is a winning committee of $\A_2$, and $W \cup D$ is a possible committee (of size $(k+t)$) for $\P_2$.
\end{proof}

\begin{lemma}\label{lem:abcu_extendk}
    Consider any $w$-\thiele rule and let $k, t \geq 1$ be two integers. In each of the three models of uncertainty, there is a polynomial-time reduction from $\poscompar{k}$ under $w$-Thiele to $\poscompar{k+t}$ under the same rule.
\end{lemma}
\begin{proof}
Assume we are given an instance of $\poscompar{k}$ under $w$-Thiele. We define an instance of $\poscompar{k+t}$ under the same rule as follows. We simply take the original instance and add $t$ dummy candidates $D = \set{d_1, \dots, d_t}$ which are all not approved by any of the original voters. Then we add enough dummy voters, all approving without uncertainty exactly the set $D$, such that $D$ is part of every winning committee under $w$-\thiele for $k+t$. 
Now let $\A$ be a completion of the original instance of $\poscompar{k}$. There exists a corresponding completion $\A'$ of the instance of $\poscompar{k+t}$ where every original voter approves exactly the candidates she approves in $\A$ (and every dummy voter approves exactly the candidates in $D$).
Then a committee $W$ of size $k$ wins in $\A$ under $w$-\thiele if and only if $W \cup D$ wins in $\A'$ under the same rule.
\end{proof}

\subsection{Three-Valued Approval}\label{sec:poscom-3va}

We now discuss the complexity of $\poscompar{k}$ in the 3VA model for different Thiele rules. We start with \av and \cc.
For \av, it turns out that in order solve the possible committee problem it suffices to consider one specific completion of the partial profile. This yields an efficient algorithm. For \cc however, it turns out that the problem is hard. We show this via a reduction from a well-studied 3SAT variant.

\begin{theorem}
\label{thm:simplePoscomkApproval}
In the 3VA model, $\poscompar{k}$ is solvable in polynomial time under Approval Voting, for all fixed $k \geq 1$.
\end{theorem}
\begin{proof}
Given a partial profile $\P = (P(1), \dots, P(n))$ and a subset $W \subseteq C$ of size $k$, define a completion $\A = (A(1), \dots, A(n))$ where for every voter $v \in V$, $\middle_\A(v) = \middle(v) \cap W$. By this construction, except for $\top(v)$, voters only approve candidates in $W$. 
We now show that $W$ is a possible committee for $\P$ if and only if $W$ is a winning committee in $\A$. One direction is immediate, since if $W$ is a winning committee in $\A$ then by definition $W$ is a possible committee.

For the other direction, assume that $W$ is a possible committee, let $\A_1 = (A_1(1), \dots, A_1(n))$ be a completion where $W$ is a winning committee. Define another completion $\A_2 = (A_2(1), \dots, A_2(n))$ where for every voter $v \in V$, $\middle_{\A_2}(v) = \middle_{\A_1}(v) \cup (\middle(v) \cap W)$. Let $v$ be a voter, and let $t = |\middle_{\A_2}(v) \setminus \middle_{\A_1}(v)|$. Observe that $w(|W \cap A_2(v)|) = w(|W \cap A_1(v)|) + t$, and for every other size-$k$ subset $W' \subseteq C$ we have $w(|W' \cap A_2(v)|) \leq w(|W' \cap A_1(v)|) + t$. We can deduce that $W$ is also a winning committee of $\A_2$.
Now, for every voter $v \in V$, remove from $\middle_{\A_2}(v)$ all candidates which are not in $W$, we get the completion $\A$ that we defined earlier. Since the number of candidates in $W$ that each voter approves is unchanged, we have $\tsc(\A, W) = \tsc(\A_2, W)$, and for every other size-$k$ subset $W'$ we have $\tsc(\A, W') \leq \tsc(\A_2, W')$. Therefore, $W$ is a winning committee of $\A$.
\end{proof}

\begin{theorem}
\label{thm:simplePoscom2CC}
In the 3VA model, $\poscompar{k}$ is NP-complete under Chamberlin--Courant, for all fixed $k \geq 2$.
\end{theorem}
\begin{proof}
Recall that $\poscompar{k}$ is in NP under all Thiele rules.
We show the hardness for $\poscompar{2}$ here and can extend this to any $k\geq 2$ using \Cref{lem:abcu_extendk}.
We show NP-hardness by a reduction from \emph{one-in-three positive 3SAT}, 
which is the following decision problem: Let $X = \set{x_1, \dots, x_n}$ be a set of elements. Given sets $S_1, \dots, S_m \subseteq X$, where $|S_j| = 3$ for every $j \in [m]$, is there a subset $B \subseteq X$ such that for each $j \in [m]$, $|B \cap S_j| = 1$? This problem is NP-complete~\citep[Problem~LO4]{GaJo79a}. 

Given $X$ and $\S = \set{S_1, \dots, S_m}$, we define an instance of $\poscompar{2}$ under Cham\-berlin--Courant. The candidates are $C = \S \cup \set{w_1, w_2}$. The partial profile $\P = \P_1 \circ \A_2$ consists of two parts that we describe next. 
The first part, $\P_1$, consists of a voter for every element in $X$. Let $x \in X$, and let $S(x) = \set{S_i \in \S : x \in S_i}$ be the sets of $\S$ that contain~$x$. Define a voter $v_x$ with $\top(v_x) = \S \setminus S(x)$, $\middle(v_x) = \set{w_1, w_2}$ and $\bottom(v_x) = S(x)$. The decision whether to approve $w_1$ or $w_2$ represents the decision whether to include~$x$ in the subset $B \subseteq X$. The second part, $\A_2$, consists of six voters without uncertainty (voters with $\middle(i) = \emptyset$). We add three voters approving $\S$, two voters approving $\set{w_1}$, and a single voter approving $\set{w_2}$.

We now show that $\set{w_1, w_2}$ is a possible committee if and only if there is a solution to the instance of one-in-three positive 3SAT.
For this we start with some observations regarding the voting profile. Let $\A = \A_1 \circ \A_2$ be a completion of $\P$. Assume that for every voter $v_x$ of $\A_1$, we have either $\middle_\A(v_x) = \set{w_1}$ or $\middle_\A(v_x) = \set{w_2}$. Observe that for every voter $v_x$ we have $A_1(v_x) \cap \set{w_1, w_2} \neq \emptyset$ hence 
\begin{align*}
    \tsc(\A, \set{w_1, w_2}) &= \tsc(\A_1, \set{w_1, w_2}) + \tsc(\A_2, \set{w_1, w_2}) = n + 3.
\end{align*}
Let $B = \set{x \in X : \middle_\A(v_x) = \set{w_1}}$ be the elements for which the corresponding voter approves $w_1$. Let $S_i \in \S$. For every $x \in S_i$, we have $\tsc(A_1(v_x), \set{w_1, S_i}) = 1$ if $x \in B$, and $\tsc(A_1(v_x), \set{w_1, S_i}) = 0$ otherwise. We have $\tsc(A_1(v_x), \set{w_1, S_i}) = 1$ for every $x \in X \setminus S_i$, since $v_x$ always approves $S_i$. Since $|S_i| = 3$, we can deduce that
\begin{align*}
    & \tsc(\A, \set{w_1, S_i}) = \tsc(\A_1, \set{w_1, S_i}) + \tsc(\A_2, \set{w_1, S_i}) = (|B \cap S_i| + n-3) + 5 \\
    &= \tsc(\A, \set{w_1, w_2}) - 1 + |B \cap S_i|.
\end{align*}
Similarly, the score of $\set{w_2, S_i}$ is
\begin{align*}
    \tsc(\A, \set{w_2, S_i}) &= (|S_j \setminus B| + n-3) + 4 = \tsc(\A, \set{w_1, w_2}) - 2 + |S_i \setminus B|.
\end{align*}
Finally, let $S_i, S_j \in \S$. There are at most $n$ voters in $\A_1$ who approve one of $S_i, S_j$, and three voters of $\A_2$ who approve $S_i$ and $S_j$. Therefore, $\tsc(\A, \set{S_i, S_j}) \leq n+3 = \tsc(\A, \set{w_1, w_2})$.

To show that $\set{w_1, w_2}$ is a possible committee if and only if there is a solution to the instance of one-in-three positive 3SAT, we do the following. Assume that there exists $B \subseteq X$ which satisfies $|B \cap S_i| = 1$ for each $i \in [m]$. Since each $S_i$ consists of three elements, we also have $|S_i \setminus B| = 2$. Define a completion $\A = \A_1 \circ \A_2$ of $\P$ where for every $x \in X$, if $x \in B$ then $\middle_\A(i_x)=  \set{w_1}$, otherwise $\middle_\A(i_x)=  \set{w_2}$. By our analysis of the scores in different completions, for every $S_i \in \S$ we have $\tsc(\A, \set{w_1, S_i}) = \tsc(\A, \set{w_2, S_i})=  \tsc(\A, \set{w_1, w_2})$, and $\tsc(\A, \set{S_i, S_j}) \leq \tsc(\A, \set{w_1, w_2})$ for every pair $S_i, S_j$. Hence, $\set{w_1, w_2}$ is a winning committee of $\A$.

Conversely, assume that there is a completion $\A = \A_1 \circ \A_2$ of $\P$ where $\set{w_1, w_2}$ is a winning committee. Let $v_x$ be a voter in $\A_1$. If $\middle_\A(v_x) = \set{w_1, w_2}$, then we can arbitrarily remove one of these candidates from $\middle_\A(v_x)$. The score of $\set{w_1, w_2}$ remains the same, and the score of any other subset of candidates cannot increase. If $\middle_\A(v_x) = \emptyset$, then we can arbitrarily add one of $w_1, w_2$ to $\middle_\A(v_x)$. The score of $\set{w_1, w_2}$ increases by 1, and the score of any other size-$k$ subset increases by at most 1. Therefore, there exists a completion $\A' = \A_1' \circ \A_2$ where $\set{w_1, w_2}$ is a winning committee, and for every voter $v_x$ of $\A_1'$, we have either $\middle_{\A'}(v_x) = \set{w_1}$ or $\middle_{\A'}(v_x) = \set{w_2}$.

As before, define $B = \set{x \in X : \middle_{\A'}(v_x) = \set{w_1}}$. For every $S_i \in \S$, we must have $\tsc(\A', \set{w_1, S_i}) \leq \tsc(\A', \set{w_1, w_2})$ hence $|B \cap S_i| \leq 1$. On the other hand, we must have $\tsc(\A', \set{w_2, S_i}) \leq \tsc(\A', \set{w_1, w_2})$ hence $|S_i \setminus B| \leq 2$. We can deduce that $|B \cap S_i| = 1$ for every $i \in [m]$.
\end{proof}

Next, we use Lemma~\ref{lemma:thieleReduction} and Theorem~\ref{thm:simplePoscom2CC} to obtain another hardness result.

\begin{theorem}
\label{thm:simplePoscomkThieleConsecutiveEqual}
Let $w$ be a weight function such that $w(k-2) < w(k-1) = w(k)$ for some $k \geq 2$. Then, in the 3VA model, $\poscompar{k}$ is NP-complete under $w$-Thiele.
\end{theorem}
\begin{proof}
Let $w_{CC}$ be the weight function for Chamberlin--Courant, namely, $w_{CC}(0) = 0$ and  $w_{CC}(x) = 1$ for every $x \geq 1$. Let $g$ be the linear function that satisfies $g(0) = w(k-2)$ and $g(1) = w(k-1)$, note that it is strictly increasing. Observe that for $t=k-2$,
\begin{align*}
    w(0+t) &= w(k-2) = g(0) = g(w_{CC}(0)) \\
    w(1+t) &= w(k-1) = g(1) = g(w_{CC}(1)) \\
    w(2+t) &= w(k) = w(k-1) = g(1) = g(w_{CC}(2))
\end{align*}

By Lemma~\ref{lemma:thieleReduction}, there exists a reduction from $\poscompar{2}$ under Chamberlin--Courant ($w_{CC}$-Thiele) to $\poscompar{2+t} = \poscompar{k}$ under $w$-Thiele. By Theorem~\ref{thm:simplePoscom2CC}, $\poscompar{2}$ is NP-complete under Chamberlin--Courant, therefore $\poscompar{k}$ is NP-complete under $w$-Thiele.
\end{proof}

Whereas the condition in Theorem~\ref{thm:simplePoscomkThieleConsecutiveEqual} does not hold for PAV (for which the complexity of $\poscompar{k}$ in the 3VA model remains open), it clearly holds for binary Thiele rules. We can again use \Cref{lem:abcu_extendk} to extend the result to larger (but fixed) values of $k$.

\begin{corollary}
\label{cor:simplePoscomkThiele1Weight}
For every binary Thiele rule, there exists \mbox{$k^* \geq 2$} such that $\poscompar{k}$ is NP-complete in the 3VA model, for all fixed $k \geq k^*$.
\end{corollary}

\subsection{Linear Incomplete Approval}\label{sec:poscom-linear}

Next, we discuss the complexity of $\poscompar{k}$ in the linear model. We start with binary Thiele rules. 

\begin{theorem}
\label{thm:linearPoscomkThiele1Weight}
In the linear model, $\poscompar{k}$ is solvable in polynomial time for binary Thiele rules, for all fixed \mbox{$k \geq 1$}.
\end{theorem}
\begin{proof}
Consider a binary Thiele rule and let $t$ be such that $w(x) = 0$ for all $x < t$ and $w(x) = 1$ for all $x \geq t$. 
Without loss of generality, we can assume that $k \geq t$ because otherwise the score of every committee of size $k$ is always zero.
Given $\P = (P(1), \dots, P(n))$ and a subset $W \subseteq C$ of $k$ candidates
we construct a completion $\A$ as follows.
For every $v \in V$, let $c_{i_1}  \succeq_v  \dots  \succeq_v  c_{i_q}$ be the linear order of voter $v$ over the candidates in $\middle(v)$.
If $|\top(v) \cap W| \geq t$, then voter $v$ always contributes a score of~1 to~$W$. Since approving more candidates cannot increase the score of $W$, we select $\middle_\A(v) = \emptyset$. If $|(\top(v) \cup \middle(v)) \cap W| < t$, then voter $v$ always contributes a score of~0 to $W$, and we select $\middle_\A(v) = \emptyset$ again. Finally, assume that $|\top(v) \cap W| < t$ and $|(\top(i) \cup \middle(v)) \cap W| \geq t$. Let $j$ be minimal with $\left|(\top(v) \cup \set{c_{i_1}, \dots, c_{i_j}}) \cap W\right| = t$ and set $\middle_\A(v) = \set{c_{i_1}, \dots, c_{i_j}}$.

We now show that $W$ is a possible committee in $\P$ if and only if $W$ is a winning committee in $\A$. One direction is immediate, since if $W$ is a winning committee in $\A$, by definition it is a possible committee. For the other direction, assume that $W$ is a possible committee, let $\A'$ be a completion where $W$ wins. Let $v \in V$ be a voter, consider three cases:
\begin{enumerate}
    \item If $|\top(v) \cap W| \geq t$ then $\tsc(A(v), W) = \tsc(A'(v), W) = 1$. For every other size-$k$ subset $W'\subseteq C$ we have $\tsc(A(v), W') \leq \tsc(A'(v), W')$. Therefore, if we change $A'(v)$ to $A(v)$, $W$ remains a winning committee.
    \item If $|\top(v) \cap W| < t$ and $|(\top(v) \cup \middle(v)) \cap W| < t$, then $\tsc(A(v), W) = \tsc(A'(v), W) = 0$, and $\tsc(A(v), W') \leq \tsc(A'(v), W')$ for every other size-$k$ subset $W'\subseteq C$. Therefore, if we change $A'(v)$ to $A(v)$, $W$ remains a winning committee.
    \item If $|\top(v) \cap W| < t$ and $|(\top(v) \cup \middle(v)) \cap W| \geq t$ then we have two options. If $|A'(v) \cap W| < t$ then $\tsc(A(v), W) = \tsc(A'(v), W) + 1$, and for every other size-$k$ subset $W' \subseteq C$ we have $\tsc(A(v), W') \leq \tsc(A'(v), W') + 1$. Otherwise, $|A'(v) \cap W| \geq t$, we get that $\tsc(A(v), W) = \tsc(A'(v), W)$, and for every other $W'$ we have $\tsc(A(v), W') \leq \tsc(A'(v), W')$. In both cases, if we change $A'(v)$ to $A(v)$, $W$ remains a winning committee.
\end{enumerate}
Overall, we get that $W$ is a winning committee of $\A$.
\end{proof}

We now turn to the complexity of $\poscompar{k}$ under non-binary $w$-Thiele rules. The next theorem states that $\poscompar{2}$ is NP-complete under $w$-Thiele for every function $w$ whose first three values are pairwise distinct. 

\begin{theorem}
\label{thm:linearPoscom2Thiele}
Let $w$ be a weight function with $w(2) > w(1) > 0$. Then, in the linear model, $\poscompar{k}$ is NP-complete under $w$-Thiele, for all fixed $k \geq 2$.
\end{theorem}
\begin{proof}
We show the hardness for $\poscompar{2}$ here and can extend this to any $k\geq 2$ using \Cref{lem:abcu_extendk}.
Without loss of generality we have $w(1) = 1$ and $w(2) = 1+x$ where $x > 0$.
We show two reductions, depending on whether $x \in (0,1]$ or $x > 1$. Both reductions are from \emph{exact cover by 3-sets} (X3C), which is the following decision problem: Given a set $U = \set{u_1, \dots, u_{3q}}$ and a collection $E = \set{e_1, \dots, e_m}$ of 3-element subsets of $U$, can we cover all the elements of $U$ using $q$ pairwise disjoint sets from $E$? This problem is NP-complete~\citep[Problem~SP2]{GaJo79a}. 

\begin{table}[t]
    \centering
    \scalebox{1}{
        \begin{tabular}{cccc}
            \toprule
            \# of voters & $\top$ & $\middle$ & $\bottom$ \\
            \midrule
            \pbox{5cm}{\,\,$m$ (one for each \\ $e = \set{u,u',u^*} \in E$)} & -- & $u \rightarrow u' \rightarrow u^* \rightarrow c$ & $\set{d,z} \cup U \setminus \set{u,u',u^*}$ \\
            $\nicefrac{q}{2}$ & $\set{z}$ & -- & $\set{c,d} \cup U$ \\
            $\nicefrac{q}{2}$ & $\set{z} \cup U$ & -- & $\set{c,d}$ \\
            $\nicefrac{q}{2}$ & $\set{d}$ & -- & $\set{c,z} \cup U$ \\
            $\nicefrac{q}{2}$ & $\set{d} \cup U$ & -- & $\set{c,z}$ \\
            1 & $\set{c,d,z}$ & -- & $U$ \\
            \bottomrule
        \end{tabular}
    }
    \caption{Overview of the incomplete approval profile used in the first part of the proof of \Cref{thm:linearPoscom2Thiele}. The left-most column depicts the number of voters with the respective incomplete preferences.}
    \label{tbl:linearPoscomThiele1}
\end{table}
We start with the case where $x \in (0,1]$.
Given $U$ and $E$, we construct an instance of $\poscompar{2}$ under $w$-Thiele. The candidate set is $C = U \cup \set{c, d, z}$. 
The partial profile $\P = \P_1 \circ \A_2$ is the concatenation of the following two parts. The profile is also depicted in  \Cref{tbl:linearPoscomThiele1}.
\begin{enumerate}
    \item The first part, $\P_1$, consists of a voter for every set in $E$. For every $e \in E$, voter $v_e$ has $\top(v_e) = \emptyset, \, \middle(v_e) = e \cup \set{c}$, and $\bottom(v_e) = (U \setminus e) \cup \set{d,z}$. The linear order on $\middle(v_e)$ is an arbitrary order that ranks $c$ last. This means that if $v_e$ approves $c$, then she also approves the candidates of~$e$. The idea is that approving $c$ indicates that $e$ is in the cover, and disapproving $c$ indicates that $e$ is not in the cover.
    \item The second part, $\A_2$, consists of $2q+1$ voters without uncertainty (voters with $\middle(i) = \emptyset$). Assume \wilog that $q$ is even. (This is possible since for odd $q$ we can simply add three more elements to $U$ and one more set to $E$ covering exactly these three new elements. This does not change the existence of an exact cover.) We add $q/2$ voters approving $\{z\}$, $q/2$ voters approving $\{z\} \cup U$, $q/2$ voters approving $\{d\}$, $q/2$ voters approving $\{d\} \cup U$, and a single voter approving $\set{c,d,z}$. 
\end{enumerate}

We analyze some properties of the profile and the scores of different candidate sets. Let $\A = \A_1 \circ \A_2$ be a completion of $\P$. Assume, for now, that for every voter $v_e$, we either have $A(v_e) = \emptyset$ or $A(v_e) = e \cup \set{c}$. For the set $\set{d,z}$, there are $2q$ voters in $\A_2$ where $|A(v) \cap \set{d,z}| = 1$ and a single voter where $|A(v) \cap \set{d,z}| = 2$. Recall that $w(1) = 1, w(2) = 1+x$, hence $\tsc(\A_2, \set{d,z}) = 2q + 1 + x$. The voters of $\A_1$ do not approve $d$ and $z$, therefore 
\[\tsc(\A, \set{d,z}) = \tsc(\A_2, \set{d,z}) = 2q + 1 + x.\]
Similarly, we have $\tsc(\A_2, \set{c,d}) = \tsc(\A_2, \set{c,z}) = q + 1 + x$. Define $B = \set{e \in E : c \in A(v_e)}$, we get that $\tsc(\A_1, \set{c,d}) = \tsc(\A_1, \set{c,z}) = |B|$, and therefore
\begin{equation}
    \tsc(\A, \set{c,d}) = \tsc(\A, \set{c,z}) 
    = |B| + q + 1 + x 
    =  \tsc(\A, \set{d,z}) + (|B|-q). \label{eq:cdFirst}
\end{equation}
Let $u \in U$, let $\deg_{B}(u)$ be the number of subsets $e \in B$ which contain $u$. Note that $\deg_{B}(u)$ is the number of voters of $\A_1$ that approve $u$. In $\A_2$, the scores of $\set{d,u}, \set{z,u}$ satisfy
\begin{align*}
    \tsc(\A_2, \set{d,u}) = \tsc(\A_2, \set{z,u}) 
    = (q+1) + \frac{q}{2}(1+x) 
    = \rpar{\frac{3}{2} + \frac{x}{2}} q + 1,
\end{align*}
and the score of $\set{c,u}$ in $\A_2$ is $\tsc(\A_2, \set{c,u}) = (q+1)$.
Then, the scores of the sets  in $\A$ are
\begin{align}
    \tsc(\A, \set{d,u}) = \tsc(\A, \set{z,u}) &= \deg_{B}(u) +  \rpar{\frac{3}{2} + \frac{x}{2}}q + 1, \label{eq:duFirst}\\
    \tsc(\A, \set{c,u}) &= |B| + \deg_{B}(u)x + \tsc(\A_2, \set{c,u}) \notag = |B| + \deg_{B}(u)x + q + 1 \notag\\
    &= \tsc(\A, \set{c,d}) + x(\deg_{B}(u)-1). \label{eq:cuFirst}
\end{align}
Lastly, for every pair $u, u' \in U$ we have $\tsc(\A_2, \set{u,u'}) = q(1+x)$.

Next, we show that $\set{c,d}$ is a possible committee for $\P$ if and only if there exists an exact cover. Let $\A = \set{A(v_e)}_{e \in E} \circ \A_2$ be a completion where $\set{c,d}$ is a winning committee. Define $B = \set{e \in E : c \in A(v_e)}$. Observe that for every voter $v_e$ such that $e \notin B$ (\ie $v_e$ does not approve $c$), we can assume that $A(v_e) = \emptyset$ (changing $A(v_e)$ to $\emptyset$ does not change the score of $\set{c,d}$ and cannot increase the score of other candidate sets). Also note that by the definition of $\P$, for every voter $v_e$ with $c \in A(v_e)$ we have $A(v_e) = e \cup \set{c}$. Hence, our assumption on $\A$ from the analysis above holds. We must have $\tsc(\A, \set{c,d}) \geq \tsc(\A, \set{d,z})$ hence $|B| \geq q$ by \Cref{eq:cdFirst}. For every $u \in U$ we must have $\tsc(\A, \set{c,d}) \geq \tsc(\A, \set{c,u})$ hence $\deg_{B}(u) \leq 1$ by \Cref{eq:cuFirst}. Overall, $B$ contains at least $q$ sets, and the degree of every $u \in U$ is at most 1, therefore $B$ is an exact cover.

Conversely, let $B$ be an exact cover, we have $|B| = q$ and $\deg_{B}(u) = 1$ for all $u \in U$. Define a completion $\A = \set{A(v_e)}_{e \in E} \circ \A_2$. For every $e \in E$, if $e \in B$ then $A(v_e) = e \cup \set{c}$, otherwise $A(v_e) = \emptyset$. Observe that by Equation~\ref{eq:cdFirst},
\[\tsc(\A, \set{c,d}) = \tsc(\A, \set{c,z}) = \tsc(\A, \set{d,z}) = 2q + 1 + x,\]
and by Equation~\ref{eq:cuFirst}, $\tsc(\A, \set{c,u}) = \tsc(\A, \set{c,d})$ for all $u \in U$. Let $u \in U$, recall that $u$ is covered by exactly one set in $B$. If $x = 1$ then by Equation~\ref{eq:duFirst} we have
\begin{align*}
    \tsc(\A, \set{d,u}) = \tsc(\A, \set{z,u}) = 1 +  \rpar{\frac{3}{2} + \frac{x}{2}}q + 1 = 2q + 2 = 2q + 1 + x = \tsc(\A, \set{c,d}).
\end{align*}
Otherwise, $x < 1$, hence $3/2 + x/2 < 2$, and for large enough $q$ we have
\begin{align*}
    \tsc(\A, \set{d,u}) = \tsc(\A, \set{z,u}) = 1 +  \rpar{\frac{3}{2} + \frac{x}{2}}q + 1 < 2q + 1 + x = \tsc(\A, \set{c,d}).
\end{align*}
Similarly, we also have  $\tsc(\A, \set{u,u'}) < \tsc(\A, \set{c,d})$ for every pair $u, u' \in U$. Overall, $\set{c,d}$ is a winning committee of $\A$, therefore it is a possible committee. This completes the proof for the case where $x \in (0,1]$.

\begin{table}[t]
    \centering
    \scalebox{1}{
        \begin{tabular}{cccc}
            \toprule
            \# of voters & $\top$ & $\middle$ & $\bottom$ \\
            \midrule
            \pbox{5cm}{\,\,$m$ (one for each \\ $e = \set{u,u',u^*} \in E$)} & -- & $u \rightarrow u' \rightarrow u^* \rightarrow c$ & $\set{d,z} \cup U \setminus \set{u,u',u^*}$ \\
            $\nicefrac{q}{x}$ & $\set{d,z}$ & -- & $\set{c} \cup U$ \\
            $\nicefrac{q}{x}$ & $U$ & -- & $\set{c,d,z}$ \\
            1 & $\set{c,d,z}$ & -- & $U$ \\
            \bottomrule
        \end{tabular}
    }
    \caption{Overview of the incomplete approval profile used in the second part of the proof of \Cref{thm:linearPoscom2Thiele}. The left-most column depicts the number of voters with the respective incomplete preferences.}
    \label{tbl:linearPoscomThiele2}
\end{table}
Next, we prove the hardness for weight functions which satisfy $x > 1$. We again use a reduction from \emph{exact three cover (X3C)} as defined above. Given $U$ and $E$, we construct an instance of $\poscompar{2}$ under $w$-Thiele. The candidate set is $C = U \cup \set{c, d, z}$. The partial profile $\P = \P_1 \circ \A_2$ is the concatenation of two parts again that we describe next. The profile is also depicted in  \Cref{tbl:linearPoscomThiele2}.
\begin{enumerate}
    \item The first part, $\P_1$, is the same as in the previous case: for every $e \in E$ we have a voter $i_e$ with $\top(i_e) = \emptyset, \middle(i_e) = e \cup \set{c}$ and $\bottom(i_e) = U \cup \set{d,z}$. The linear order on $\middle(i_e)$ is an arbitrary order that ranks $c$ last. Analogous to the previous case, the idea is that approving $c$ indicates that $e$ is in the cover, and disapproving $c$ indicates that $e$ is not in the cover.
    \item Assume \wilog that $q / x$ is an integer (as we already argued in the first part of the proof, we can always add a constant number of elements to $U$ and $E$ such that $q$ will be a multiple of $x$, without changing the existence of an exact cover). The second part $\A_2$ consists of $(2q/x)+1$ voters without uncertainty. We have $q/x$ voters approving $\set{z,d}$, $q/x$ voters approving $U$, and a single voter approving $\set{c,d,z}$. This completes the construction of the profile.
\end{enumerate}

Let $\A = \A_1 \circ \A_2$ be a completion of $\P$. As in the proof for the previous case, assume that for every voter $v_e$ either $A(v_e) = \emptyset$ or $A(v_e) = e \cup \set{c}$. We analyze the scores of different candidate sets under $\A$. For the set $\set{d,z}$, there are $(q/x)+1$ voters in $\A_2$ where $|A(i) \cap \set{d,z}| = 2$, and the other voters do not approve $z,d$. The voters of $\A_1$ do not approve $d$ and $z$. Since $w(1) = 1, w(2) = 1+x$, we have
\begin{align*}
    \tsc(\A, \set{d,z}) &= \tsc(\A_2, \set{d,z}) = \rpar{\frac{q}{x}+1}(1+x) = q \rpar{\frac{1}{x} + 1} + 1 + x.
\end{align*}
Similarly, we have
\[\tsc(\A_2, \set{c,d}) = \tsc(\A_2, \set{c,z}) = \frac{q}{x} + 1 + x.\]
Define $B = \set{e \in E : c \in A(v_e)}$ as above. Then again, $\tsc(\A_1, \set{c,d}) = \tsc(\A_1, \set{c,z}) = |B|$,
hence 
\begin{align}
    \tsc(\A, \set{c,d}) &= \tsc(\A, \set{c,z}) = \frac{q}{x} + 1 + x + |B| = \tsc(\A, \set{d,z}) + |B|-q. \label{eq:cdSecond}
\end{align}
Let $u \in U$, let $\deg_{B}(u)$ be the number of sets in $B$ which cover $u$. Note that $\deg_{B}(u)$ is the number of voters of $\A_1$ that approve $u$. The scores of $\set{d,u}, \set{z,u},$ and $\set{c,u}$ are as follows. Recall that $x > 1$ therefore $2 / x < 1 + 1/x$.
\begin{align}
    \tsc(\A, \set{d,u}) &= \tsc(\A, \set{z,u}) \notag = \deg_{B}(u) + \rpar{\frac{2q}{x}+1} \notag= \deg_{B}(u) + q \cdot \frac{2}{x} + 1 \notag\\
    &< \deg_{B}(u) + q\rpar{\frac{1}{x} + 1} + 1, \label{eq:duSecond}\\
    \tsc(\A, \set{c,u}) &= \rpar{\frac{q}{x} + 1} + |B| + \deg_{B}(u) \cdot x = \tsc(\A, \set{c,d}) + x \cdot (\deg_{B}(u) - 1). \label{eq:cuSecond}
\end{align}
Lastly, for every pair $u, u' \in U$ we have $\tsc(\A_2, \set{u,u'}) = \tsc(\A, \set{d,z})-1-x$.

From Equations~\ref{eq:cdSecond}, \ref{eq:duSecond} and \ref{eq:cuSecond} we can deduce that $\set{c,d}$ is a possible committee for $\P$ if and only if there exists an exact cover, by a very similar argument to the proof of the previous case where $x \in (0,1]$.
\end{proof}

Next, we use \Cref{lemma:thieleReduction} in order to generalize \Cref{thm:linearPoscom2Thiele} to every function with three consecutive different values.

\begin{theorem}
\label{thm:linearPoscomkThieleConsecutiveDifferent}
Let $w$ be a weight function such that $w(k-2) < w(k-1) < w(k)$ for some $k \geq 2$. Then, in the linear model, $\poscompar{k}$ is NP-complete under $w$-Thiele.
\end{theorem}
\begin{proof}
    Define a weight function $w'$ where $w'(x) = w(x+k-2) - w(k-2)$ for every $x \geq 0$. In particular, $w'(0) = 0, w'(1) = w(k-1)-w(k-2) > 0$ and $w'(2) = w(k)-w(k-2) > w'(1)$. By \Cref{thm:linearPoscom2Thiele}, $\poscompar{2}$ is NP-complete under $w'$-Thiele. By setting $t = k-2$ and a function $g(x) = x + w(k-2)$ we get $w(x+t) = g(w'(x)) = w'(x) + w(k-2)$ for every $x \in \set{0, 1, 2}$.   
    By \Cref{lemma:thieleReduction}, there is a reduction from $\poscompar{2}$ under $w'$-Thiele to $\poscompar{k}$ under $w$-Thiele. Therefore, $\poscompar{k}$ is NP-complete under $w$-Thiele.
\end{proof}

Clearly, Theorem~\ref{thm:linearPoscomkThieleConsecutiveDifferent} applies to all Thiele rules with strictly increasing weight functions (such as AV and PAV). 

\begin{corollary}
\label{cor:linearPoscomkThieleIncreasing}
Let $w$ be a strictly increasing weight function. Then, in the linear model, $\poscompar{k}$ is NP-complete under $w$-Thiele for all fixed $k \geq 2$.
\end{corollary}

As a special case, Corollary~\ref{cor:linearPoscomkThieleIncreasing} reasserts the result of \citet{BGL+13a} that $\poscompar{k}$ is NP-complete for the AV rule. (The hardness result by \citet{BGL+13a} even applies to the special case in which all candidates are unknown.)

\subsection{Poset Approval}\label{sec:poscom-proof}
Combining the results from the previous two sections, we can now formulate the main result of this section. 

\begin{theorem}
\label{thm:posetPoscomkThiele}
In the poset approval model, for every weight function $w$ there exists an integer $k^*$ such that $\poscompar{k}$ is NP-complete under $w$-Thiele, for all fixed $k \geq k^*$.
\end{theorem}

\begin{proof}
Let $w$ be a weight function, and let $j$ be the minimal index such that $w(j) > 0$. Since $w$ is non-decreasing, either $w(j+1) = w(j) > w(j-1)$ or $w(j+1) > w(j) > w(j-1)$. In the first case, by Theorem~\ref{thm:simplePoscomkThieleConsecutiveEqual}, $\poscompar{j+1}$ is NP-complete under $w$-Thiele in the 3VA model. In the second case, by Theorem~\ref{thm:linearPoscomkThieleConsecutiveDifferent},  $\poscompar{j+1}$ is NP-complete under $w$-Thiele in the linear model. We can extend the hardness in both cases to any $k\geq j+1$ by using \Cref{lem:abcu_extendk} again.
The result follows since 3VA and the linear model are special cases of poset approval. 
\end{proof}

\subsection{Possible Committee Members}
We now consider a different decision problem where we do not focus on whole committees but rather on individual candidates.
We say candidate $c$ is a \e{possible committee member} if there exists a completion $\A$ of $\P$ such that there exists a winning committee $W$ with $c \in W$. For an ABC rule $r$ and a committee size $k$, we consider the computational problem $\posmempar{k}$, where the input consists of a partial profile $\P$ and a candidate $c$, and the goal is to determine whether $c$ is a possible committee member under $r$.

Note that for every voting rule, if $\poscompar{k}$ is solvable in polynomial time, then we can also solve $\posmempar{k}$ efficiently: given a candidate $c$, test whether $\set{c, d_1, \dots, d_{k-1}}$ is a possible committee for all other $k-1$ candidates $d_1, \dots, d_{k-1}$. Since $k$ is fixed, this can be done efficiently.
By Theorems~\ref{thm:simplePoscomkApproval} and~\ref{thm:linearPoscomkThiele1Weight}  we can deduce the following.

\begin{corollary}
For all fixed $k \geq 1$, $\posmempar{k}$ is solvable in polynomial time:
\begin{enumerate}
\item In the 3VA model under approval voting;
\item In the linear model under every binary Thiele rule.
\end{enumerate}
\end{corollary}

We also obtain a tractability result in a setting for which the corresponding $\poscompar{k}$ problem is intractable. 

\begin{theorem}
\label{thm:linearPosmemkApproval}
In the linear model, $\posmempar{k}$ is solvable in polynomial time under Approval Voting, for all fixed $k \geq 1$. 
\end{theorem}
\begin{proof}
Given a candidate $c$ and a partial profile $\P = (P(1), \dots, P(n))$, we construct a completion $\A$ of $\P$ as follows. 
For all voters $v$ with  $c \in \top(v)$ or $c \in \bottom(v)$, we let $\middle_\A(v) = \emptyset$. Otherwise, $c \in \middle(v)$ and we define $\middle_\A(v)$ to be the smallest prefix of the linear order over $\middle(v)$ that includes $c$. 
We show that $c$ is a possible committee member if and only if $c$ is a member of a committee in $\A$.

Let $\A'$ be a completion where $c$ is a committee member. Recall from  \Cref{sec:abcu_prelims} that AV can also be defined using scores of individual candidates. By this definition, there are at most $k-1$ candidates with a score greater than $c$ in $\A'$. Let $v$ be a voter, consider three cases. First, if $c \in \top(v) \cup \bottom(v)$, then $\middle_\A(v) = \emptyset$, hence $s(A(v), c) = s(A'(v), c)$ and $s(A(v), d) \leq s(A'(v), d)$ for every other candidate $d \in C$. Second, if $c \in \middle(v)$ and $c \in A'(v)$, then $\middle_\A(v) \subseteq \middle_{\A'}(v)$, hence we also have $s(A(v), c) = s(A'(v), c)$ and $s(A(v), d) \leq s(A'(v), d)$ for every other candidate $d$. Finally, if $c \in \middle(v)$ and $c \notin A'(v)$, then $s(A(v), c) = s(A'(v), c) + 1$ and $s(A(v), d) \leq s(A'(v), d) + 1$ for every other candidate $d$. We can deduce that in $\A$ there are also at most $k-1$ candidates with a score greater than $c$, therefore $c$ is a committee member in $\A$.
\end{proof}

\section{Computing Necessary Committees}
\label{sec:neccom}

The next decision problem we study concerns necessary winning committees.
For an ABC rule $\rrule$ and a (fixed) committee size $k$, we consider the following decision problem, which we refer to as $\neccompar{k}$: 
Given a partial profile $\P$ and a subset $W\subseteq C$ of candidates of size $|W|=k$, 
determine whether $W$ is a necessary committee under $\rrule$.
As above, we parametrize the problem by the committee size $k$ to evade hardness even in the case of complete information.

We show that $\neccompar{k}$ is solvable in polynomial time for all \abc scoring rules in all three models of incompleteness. 
Each member of the family of ABC scoring rules is associated with a \emph{scoring function} $f \colon \naturals \times \naturals \rightarrow \reals$ satisfying $f(x,y) \geq f(x', y)$ for $x \geq x'$. For a voter $v$ and subset $S \subseteq C$, the score of $S$ given approval set $A(v)$ is 
\[\tsc(A(v), S) = f(|A(v) \cap S|, |A(v)|),\] 
and the total score is $\tsc(\A, S) = \sum_{v \in V} \tsc(A(v), S)$. The rule defined by $f$ outputs the subset(s) of size $k$ with the highest score. 
Observe that \thiele rules are a special case of ABC scoring rules, where the score $\tsc(A(v), S)$ that $v$ assigns to $S$ depends only on $|A(v) \cap S|$. An example for an ABC scoring rule which is not a Thiele rule is \e{Satisfaction Approval Voting} \citep{BrKi14a}, where the score a voter $v$ contributes to a $S \subseteq C$ is $\tsc(A(v), S) = |S \cap A(v)| / |A(v)|$. 
When we discuss scoring rules, we assume that $f(x,y)$ is computable in polynomial time given $x$ and $y$. 

\begin{theorem}
\label{thm:posetNeccomkThiele}
Let $\rrule$ be an ABC scoring rule. Then, in the model of poset approval, for all fixed $k \geq 1$, $\neccompar{k}$ is solvable in polynomial time under $\rrule$.
\end{theorem}

In order to prove \Cref{thm:posetNeccomkThiele}, observe that a subset $W \subseteq C$ of size $k$ is \emph{not} a necessary committee if and only if there exists a completion where another size-$k$ subset has a score strictly greater than the score of $W$. Formally, let $\P$ be a partial profile 
and let $\Ck = \set{W' \subseteq C : |W'| = k}$ be the set of subsets of size $k$. For a completion $\A$ of $\P$ and $W, W' \in \Ck$, define the \textit{score difference} as
\[\D_\A(W, W') = \tsc(\A, W') - \tsc(\A, W) \text.\] 
When $\A$ consists of a single approval set $A(v)$ we write $\D_{A(v)}(W, W')$ instead of $\D_\A(W, W')$.

Let $\comp(\P)$ be the set of completions of $\P$. Again, when the profile consists of a single voter $v$, we write $\comp(P(v))$ instead of $\comp(\P)$.  Define the \emph{maximal score difference} as
\[
    \MD_\P(W) = \max_{\A \in \comp(\P), W' \in \Ck} \D_\A(W, W').
\]
Observe that $W$ is \e{not} a necessary committee if and only if $\MD_\P(W) > 0$. To compute $\MD_\P(W)$, we iterate over all sets $W' \in \Ck$, and for every $W'$ we compute 
\[ \MD_\P(W, W') = \max_{\A \in \comp(\P)} \D_\A(W, W').\] 
Once we compute $\MD_\P(W, W')$ for every $k$-set of candidates $W' \in \Ck$ we are done, since $\MD_\P(W) = \allowbreak \max_{W' \in \Ck} \MD_\P(W, W')$.

Now, let $W' \in \Ck$. Since $\rrule$ is an ABC scoring rule, the score of $W'$ in a completion $\A$ is $\tsc(\A, W') = \sum_{v \in V} \tsc(A(v), W')$, therefore 
\begin{align*}
    & \MD_\P(W, W') = \max_{\A \in \comp(\P)} \sum_{v \in V} \D_{A(v)}(W, W') 
    = \sum_{v \in V} \max_{A(v) \in \comp(P(v))} \D_{A(v)}(W, W') 
    =  \sum_{v \in V} \MD_{P(v)}(W, W') \,.
\end{align*}

The following lemma completes the proof of \Cref{thm:posetNeccomkThiele} by showing that $\MD_{P(v)}(W, W')$ can be computed in polynomial time for each voter $v \in V$.

\begin{lemma}
Let $\rrule$ be an ABC scoring rule and $k$ a fixed natural number. Computing $\MD_{P(v)}(W, W')$ is possible in polynomial time given a partial vote $P(v)$ over a set $C$ of candidates, and a pair of sets $W, W' \in \Ck$.
\end{lemma}
\begin{proof}
Set $S = \middle(v) \cap (W \cup W')$. We call a subset $R \subseteq S$ \emph{feasible} if there exists a completion $A(v)$ of $P(v)$ where $\middle_\A(v) \cap S = R$. For feasible $R \subseteq S$, set
\begin{align*}
    \MD_{P(v)}^R(W, W') = \max_{A(v) \in \comp(P(v)) :\middle_A(v) \cap S = R} \D_{A(v)}(W, W') \,.
\end{align*}
Otherwise, define $\MD_{P(v)}^R(W, W') = -\infty$. With that we obtain 
\[\MD_{P(v)}(W, W') = \max_{R \subseteq S} \MD_{P(v)}^R(W, W').\] 
Since $|S| \leq 2k$, there is a constant number of subsets of $S$. 
In the remainder of the proof we show that we can compute $\MD_{P(v)}^R(W, W')$ in polynomial time for every $R \subseteq S$.

Let $R \subseteq S$, and let $R' \subseteq \middle(v)$ be the candidates that voter $v$ necessarily approves whenever it approves the candidates of $R$. Formally, $c \in R'$ if $c \in \middle(v)$ and there exists a candidate $c^* \in R$ such that $c \succeq_v c^*$. Note that $R \subseteq R'$. If $R' \cap (S \setminus R) \neq \emptyset$ then there exists a candidate in $S \setminus R$ that voter $v$ must approve when it approves the candidates of $R$. In this case, $R$ is not feasible, and $\MD_{P(v)}^R(W, W') = -\infty$. Otherwise, every candidate in $R'$ is either in $R$ or not in $S$, and $R$ is feasible. We now focus on completions where $\middle_A(v) \cap S = R$.

Let $T$ be the set of candidates $c \in \middle(v) \setminus (S \cup R')$ such that there is no $c' \in S \setminus R$ for which $c' \succeq_v c$. These are the remaining candidates that voter $v$ can approve without approving additional candidates from $S$. Note that the approval of every candidate outside of $T$ is already decided, since we focus on completions where $\middle_A(v) \cap S = R$.

Recall that $\rrule$ is a scoring rule, hence the scores $\tsc(A(v), W), \tsc(A(v), W')$ depend only on $|A(v) \cap W|, |A(v) \cap W'|$ and $|A(v)|$. Since $S = \middle(v) \cap (W \cup W')$, we already determined $|A(v) \cap W|$ and $|A(v) \cap W'|$. To determine $|A(v)|$ we only need to decide how many candidates we approve from $T$. We can deduce that among the completions where $\middle_A(v) \cap S = R$, only the number of candidates that voter $v$ approves from $T$ can affect the scores of $W$ and $W'$.

Let $c_1, \dots, c_t$ be a topological sorting of the candidates of $T$ \wrt $\succeq_v$ (\ie for every pair $c_i \neq c_j \in T$, if $c_i \succeq_v c_j$ then $i < j$). For every $i \in [t]$ define a completion $A_i(v) = \top(v) \cup R' \cup \set{c_1, \dots, c_i}$. It is easy to verify that for every $i \in [t]$, $A_i(v)$ is indeed a completion of $P(v)$ and it satisfies $\middle_{A_i}(v) \cap S = R$. Finally, by the observation regarding $\rrule$ and the scores $\tsc(A(v), W), \tsc(A(v), W')$, we get that
\begin{align*}
    \MD_{P(v)}^R(W, W') &= \max_{A(v) \in \comp(P(v)) : \middle_A(v) \cap S = R} \D_{A(v)}(W, W') = \max_{i \in [t]} \D_{A_i(v)}(W, W').
\end{align*}
We can deduce that we can compute $\MD_{P(v)}^R(W, W')$ in polynomial time, by computing the score of $W, W'$ in $A_1(v), \dots, A_t(v)$, which completes the proof.
\end{proof}

\subsection{Necessary Committee Members}
A candidate $c$ is a \e{necessary committee member} if for every completion  $\A$ of $\P$ there exists a winning committee $W$ for which $c \in W$. For an ABC rule $r$ and a committee size $k$, we consider the computational problem $\necmempar{k}$, where the input consists of a partial profile $\P$ and a candidate $c$, and the goal is to determine whether $c$ is a necessary committee member.

Note that a necessary committee member can be a member of different winning committees in different completions. Thus, in contrast to $\poscom$ and $\posmem$, being able to solve $\neccom$ does not help deciding $\necmem$, because a candidate could be a necessary committee member without being a member of a necessary committee. Furthermore, a candidate could be a necessary committee member even if a necessary committee does not exist.

We show that in the 3VA model we can  find the necessary members in polynomial time under approval voting.

\begin{theorem}
\label{thm:simpleNecmemkApproval}
In the 3VA model, $\necmempar{k}$ is solvable in polynomial time under Approval Voting, for all fixed $k \geq 1$.
\end{theorem}
\begin{proof}
Let $\P = (P(1), \dots, P(n))$ be a partial profile over a set $C$ of candidates, and let $c$ be a candidate. For a subset $W \subseteq C \setminus \set{c}$ of size $k$ and a completion $\A$, we say that $W$ \e{defeats} $c$ in $\A$ if $s(\A, W) > s(\A, W')$ for every size-$k$ subset $W' \subseteq C$ which includes $c$. Observe that $c$ is \e{not} a necessary committee member if and only if there exists a set $W \subseteq C \setminus \set{c}$ and a completion $\A$ where $W$ defeats $c$. For every $W$ of size $k$ (which does not contain $c$), we search for a completion where $W$ defeats $c$.

Given a set $W$, define a completion $\A = (A(1), \dots, A(n))$ where for every voter $v$, $\middle_\A(v) = \middle(v) \cap W$. We show that $W$ defeats $c$ in $\A$ if and only if there exists a completion where $W$ defeats $c$. Assume there exists a completion $\A'$ where $W$ defeats $c$, we show that this property also holds in $\A$.
As in the proof of \Cref{thm:simplePoscomkApproval}, we can perform changes that transform $\A'$ to $\A$, while maintaining the property that $W$ defeats $c$. We can deduce that $W$ defeats $c$ in $\A$.
\end{proof}

In the linear model, $\necmempar{k}$ is solvable in polynomial time under AV and under binary Thiele rules.

\begin{theorem}
\label{thm:linearNecmemkApproval}
In the linear model, $\necmempar{k}$ is solvable in polynomial time under Approval Voting, for all fixed \mbox{$k \geq 1$}.
\end{theorem}
\begin{proof}
Given a partial profile $\P = (P(1), \dots, P(n))$ and a candidate $c$, we construct a completion $\A$ of $\P$ as follows. For every voter $v$, we have a partition $\top(v), \middle(v), \bottom(v)$ and a linear order $c_{i_1} \succ \dots \succ c_{i_q}$ over the candidates in $\middle(v)$. If $c \in \top(v)$ or $c \in \bottom(v)$ define $\middle_\A(v) = \middle(v)$. Otherwise, $c \in \middle(A)$, let $j$ be the rank of $c$ in the linear order $c_{i_1} \succ \dots \succ c_{i_q}$, i.e., $c_{i_j} = c$. Define $\middle_\A(v) = \set{c_{i_1}, \dots, c_{i_{j-1}}}$. We show that $c$ is \e{not} a necessary committee member if and only if $c$ is not a member of a committee in $\A$.

Let $\A'$ be a completion where $c$ is not a committee member, i.e., there are at least $k$ candidates with a score greater than $c$ in $\A'$. (We again use the definition of AV using scores of individual candidates, as in the proof of \Cref{thm:linearPosmemkApproval}.) Let $v$ be a voter, consider three cases. First, if $c \in \top(v) \cup \bottom(v)$, then $s(A(v), c) = s(A'(v), c)$ and $s(A(v), d) \geq s(A'(v), d)$ for every other candidate $d \in C$. Second, if $c \in \middle(v)$ and $c \in A'(v)$, then $s(A(v), c) = s(A'(v), c) - 1$ and $s(A(v), d) \geq s(A'(v), d) - 1$ for every other candidate $d$.  Finally, if $c \in \middle(v)$ and $c \notin A'(v)$, then $s(A(v), c) = s(A'(v), c)$ and $s(A(v), d) \geq s(A'(v), d)$ for every other candidate $d$. We can deduce that in $\A$ there are also at least $k$ candidates with a score greater than $c$, therefore $c$ is not a committee member in $\A$.
\end{proof}

\begin{theorem}
\label{thm:linearNecmemkBinary}
In the linear model, $\necmempar{k}$ is solvable in polynomial time under Approval Voting and under every binary Thiele rule, for all fixed \mbox{$k \geq 1$}.
\end{theorem}
\begin{proof}
Let $\P = (P(1), \dots, P(n))$ be a partial profile over a set $C$ of candidates, and let $c$ be a candidate. As in the proof of \Cref{thm:simpleNecmemkApproval}, for every $W \subseteq C$ of size $k$ (which does not contain $c$), we search for a completion where $W$ defeats $c$.

Recall that there exists an integer $t$ such that $w(x) = 0$ for every $x < t$ and $w(x) = 1$ for every $x \geq t$. Given a size-$k$ set of candidates $W$, we construct a completion $\A$ in a similar way to the proof of Theorem~\ref{thm:linearPoscomkThiele1Weight}. Let $v$ be a voter, we have three sets $\top(v), \middle(v), \bottom(v)$ and a linear order $c_{i_1} \succ \dots \succ c_{i_q}$ over the candidates in $\middle(v)$. If $|\top(v) \cap W| \geq t$ or $|(\top(v) \cup \middle(v)) \cap W| < t$, then define $\middle_\A(v) = \emptyset$. Otherwise, $|\top(v) \cap W| < t$ and $|(\top(v) \cup \middle(v)) \cap W| \geq t$. Let $j$ be the minimal index such that $|(\top(v) \cup \set{c_{i_1}, \dots, c_{i_j}}) \cap W| \geq t$, define $\middle_\A(v) = \set{c_{i_1}, \dots, c_{i_j}}$.

We show that $W$ defeats $c$ in $\A$ if and only if there exists a completion where $W$ defeats $c$. Assume that $W$ defeats $c$ in some completion $\A'$. As in the proof of \Cref{thm:linearPoscomkThiele1Weight}, we can perform changes that transform $\A'$ to $\A$, while maintaining the property that $W$ defeats $c$. We can deduce that $W$ defeats $c$ in $\A$.
\end{proof}

\section{Representation under Incomplete Information}
\label{sec:representation}

An important goal in committee voting concerns the proportional representation of voters.
In this section, we investigate the problem of representation under incomplete preference information: Given a committee $W$ and partial profile $\P$, does $W$ possibly, or necessarily, provide proportional representation \wrt a completion of $\P$? We formalize (proportional) representation via the well-known ``justified representation'' axioms. In particular, we consider the basic axiom known as \emph{justified representation} (\emph{JR})~\cite{ABC+16a}, the related stronger representation axioms \emph{proportional justified representation} (\emph{PJR})~\cite{SFF+17a} and \emph{extended justified representation} (\emph{EJR})~\cite{ABC+16a}, and their recently introduced robust variations PJR+ and EJR+~\cite{BrPe23a}. 
These axioms share the same idea: every group of voters that is sufficiently large and cohesive should be adequately represented in the winning committee. The differences between the axioms are the exact requirements of size, cohesiveness and representation for these sets.

The axioms JR, PJR+, and EJR+ can be verified efficiently even if the committee size $k$ is part of the input \citep{ABC+16a,BrPe23a}. Therefore, this section focuses on these three axioms, and does not make the assumption that $k$ is fixed. The axioms PJR and EJR, for which the verification problem is intractable when $k$ is part of the input \citep{ABC+16a,AEH+18a}, are discussed in \Cref{sec:pjr-ejr}.

\subsection{Justified Representation}\label{sec:jr}
For a (complete) approval profile $\A$ and committee size~$k$, a committee $W\subseteq C$ of size $|W|=k$ satisfies \textit{justified representation (JR)} \wrt $\A$ if for all voter groups $G \subseteq V$ with $|G|\geq n/k$ and $\bigcap_{v \in G} A(v)\neq \emptyset$, it holds that $W \cap \bigcup_{v\in G} A(v) \neq \emptyset$.
Using the structure described before, in this section we call a group of voters $G \subseteq V$ \emph{large} if $|G|\geq n/k$, \emph{cohesive} if $\bigcap_{v \in G} A(v)\neq \emptyset$, and we say that $G$ is \emph{represented} by $W$ if $W \cap \bigcup_{v\in G} A(v) \neq \emptyset$.
\citet{ABC+16a} show that it can be tested in polynomial time whether a given committee provides JR \wrt a given (complete) approval profile.

Given a partial profile $\P$ and a committee $W$, we consider two decision problems.
In $\posjr$, we ask whether there exists a completion $\A$ of $\P$ where $W$ satisfies JR. 
In $\necjr$ we ask whether $W$ satisfies JR \wrt every completion of $\P$.
In contrast to the decision problems we considered earlier, the committee size $k$ is part of the input for both problems, that is, we do \emph{not} assume by default that $k$ is fixed.

\begin{table}[t]
    \centering
    \scalebox{1}{
        \begin{tabular}{cccc}
            \toprule
            voter $v$ & $\top(v)$ & $\middle(v)$ & $\bottom(v)$ \\
            \midrule
            1 & $a, b$ & $d \rightarrow c$ & $e$ \\
            2 & $a$ & $d \rightarrow b \rightarrow c$ & $e$ \\
            3 & $a, b$ & $c \rightarrow d \rightarrow e$ & -- \\
            4 & $a$ & $d \rightarrow b$ & $c,e$ \\
            5 & $c, e$ & -- & $a,b,d$ \\
            6 & $e$ & $c \rightarrow d \rightarrow a$ & $b$ \\
            \bottomrule
        \end{tabular}
    }
    \caption{The incomplete approval profile from \Cref{ex:posJR_necJR,ex:posEJR_necEJR}.}
    \label{tbl:example_posJR_necJR}
\end{table}

\begin{example}\label{ex:posJR_necJR}
    Consider the instance depicted in \Cref{tbl:example_posJR_necJR} using the model of linear incomplete approvals. There are 6 voters with incomplete approval preferences over 6 candidates, and we assume that the committee size is $k = 3$. 
    The committee $W=\set{b,d,e}$ does not necessarily satisfy JR: In the completion where all voters only approve candidates in $\top(v)$, voters 2 and 4 form a large and cohesive group that is not represented. 
    However, $W$ possibly satisfies JR: When voter~$4$ approves of $a$ and $d$ and all other voters only approve candidates in $\top(v)$, the group $\{2,4\}$ is represented, as are all other large and cohesive groups.   
    Furthermore, the committee $\set{a,c,e}$ necessarily satisfies JR, since in every possible completion every voter is represented at least once.
\end{example}

We first prove that a committee satisfying JR \wrt an approval profile $\A$ also satisfies JR \wrt a new approval profile $\A'$ that results from modifying $\A$ in certain ways. This allows us to then devise an efficient algorithm to solve $\posjr$ and $\necjr$. Note that the same kind of ``monotonicity'' is not true for the axioms we discuss in \Cref{sec:proportional-extended}. 

\begin{lemma}
\label{lem:JRchangeprofile}
Let $W \subseteq C$ be a committee satisfying JR \wrt an approval profile $\A$. Then, $W$ satisfies JR \wrt a modified approval profile $\A'$ if $\A'$ is constructed in one of the following two ways:
\begin{itemize}
    \item[(i)] a voter $v^* \in V$ stops to approve a candidate not in $W$, \ie $A'(v^*) = A(v^*) \setminus \set{c}$ for some $c \notin W$ and $A'(v) = A(v)$ for all $v \neq v^*$; or 
    \item[(ii)] a voter $v^* \in V$ changes her approval set such that the modified approval set contains a candidate in $W$, \ie $A'(v^*) \cap W \neq \emptyset$ and $A'(v) = A(v)$ for all $v \neq v^*$.
\end{itemize}
\end{lemma}
\begin{proof}
    To prove the two claims, assume that committee $W$ satisfies JR \wrt $\A$ and let $v^* \in V$ be the voter who changes her ballot when switching from $\A$ to $\A'$. For \textit{(i)}, we first consider all large groups of voters $G \subseteq V$ that include $v^*$. If $G$ is large and cohesive under $\A$, then we know that $W \cap (\bigcup_{v \in G} A(v)) \neq \emptyset$, \ie that $G$ is represented by $W$ under $\A$. The same is true under $\A'$ as $A(v) \cap W$ did not change for any voter when switching from $\A$ to $\A'$.
    If $V$ was not cohesive under~$\A$, it cannot become cohesive under $\A'$ through deletion of a candidate from an approval set. Thus, $W$ satisfies JR under $\A'$.
    Now consider all large groups $G \subseteq V$ that do not include~$v^*$. For these subsets, nothing changes when switching from $\A$ to $\A'$. $G$ is cohesive under $\A$ if and only if it is cohesive under $\A'$. Since $W$ satisfies JR under $\A$, $G$ is represented by $W$ under $\A$ if it is cohesive. It follows directly that $G$ is also represented by $W$ under $\A'$ in this case.
    
   For \textit{(ii)}, we again start by considering all large voter groups $G \subseteq V$ that include $v^*$. Since $A'(v^*) \cap W \neq \emptyset$, we know that $G$ is represented by $W$ (whether $G$ is cohesive or not). 
    For voter groups not including $v^*$, as in \textit{(i)}, nothing changes when switching from $\A$ to $\A'$. In particular, any large and cohesive voter group not including $v^*$ is represented by $W$ under $\A'$ because it is represented by $W$ under $\A$.
\end{proof}

Using this lemma, we show tractability of $\posjr$ and $\necjr$ by proving that for each problem it is sufficient to check whether $W$ satisfies JR for one carefully chosen completion.

\begin{theorem}
\label{thm:posetPosNeccomkJR}
In the poset approval model, $\posjr$ and $\necjr$ are both solvable in polynomial time. 
\end{theorem}
\begin{proof}
    We start with the proof for $\posjr$.
    Assume we are given a partial profile $\P = (P(1) ,\dots, P(n))$ and a committee $W$ of size $k$. We construct a completion $\A$ of $\P$ as follows. For each voter $v \in V$ with $\middle(v) \cap W \neq \emptyset$ we define $\middle_\A(v) = \middle(v)$. For all other voters $v$, we set $\middle_\A(v) = \emptyset$.
    We claim that there exists a completion of $\P$ in which $W$ satisfies JR if and only if $W$ satisfies JR in $\A$.
    
    If $W$ satisfies JR \wrt $\A$, we have a ``yes'' instance by definition. Now assume that $W$ satisfies JR in some completion $\A' = (A'(1), \ldots, A'(n))$. We argue why, in this case, $W$ also satisfies JR in $\A$. To this end, we change the approval ballots of the voters one by one from $\A'$ to $\A$ while keeping track of the compliance of $W$ with JR.
    For every voter $v \in V$ with $\middle(v) \cap W = \emptyset$, we can delete all voters in $\middle(v)$ from $A'(v)$ to obtain $A(v)$. By \Cref{lem:JRchangeprofile} (i), $W$ still satisfies JR on the resulting profile.
    For every voter $v \in V$ with $\middle(v) \cap W \neq \emptyset$ we can add all remaining candidates from $\middle(v)$ to $A'(v)$ to obtain $A(v)$. By \Cref{lem:JRchangeprofile} (ii), $W$ again still satisfies JR on the resulting profile.
    Changing the approval ballots of all voters like this will transform the profile $\A'$ to $\A$. Thus, if $W$ satisfies JR on $\A'$, it also satisfies JR on $\A$. As a result, it is sufficient to consider the completion $\A$ when checking if $W$ is a possible JR committee.

    We now consider the problem $\necjr$ and use a similar argument as above.
    Again, assume we are given a partial profile~$\P$ and a committee $W$. We construct a completion $\A$ of $\P$ as follows. For $v \in V$, let $S_v$ be the inclusion-maximal subset of candidates in $\middle(v)$ satisfying $S_v \cap W = \emptyset$ while not violating the poset ordering given over $\middle(v)$. We can construct $S_v$ efficiently by checking each candidate in $\middle(v)$ on its own: If the candidate can be approved without needing to approve a candidate from $W$, we add it to $S_v$, otherwise we don't. Then we set $\middle_\A(v) = S_v$ for every $v \in V$.
    We claim that $W$ satisfies JR \wrt every completion of $\P$ if and only if $W$ satisfies JR \wrt $\A$.
    
    If $W$ fails to satisfy JR on $\A$, we have a ``no'' instance by definition. Now assume that $W$ satisfies JR on $\A$. We prove that it then also satisfies JR for any other completion $\A' = (A'(1), \ldots, A'(n))$ of $\P$. To do this, we consider the approval set $A(v)$ of each voter $v \in V$ individually and keep track of the compliance of $W$ with JR while changing $A(v)$ to $A'(v)$.
    For every $v \in V$, we delete all candidates $A(v) \setminus A'(v)$ from $A(v)$. By construction, these are all candidates not in $W$ and thus by \Cref{lem:JRchangeprofile} (i) the committee $W$ still satisfies JR afterwards. 
    Then, we add all remaining candidates from $\middle_{\A'}(v)$ to the approval set of $v$. By the definition of $S_v$ and the poset order, this set is either empty or includes at least one $c \in W$ and thus, by \Cref{lem:JRchangeprofile} (ii), the committee $W$ still satisfies JR afterwards.
    Changing the approval ballots of all voters like this will transform the profile $\A$ to $\A'$. As a result, $W$ satisfies JR \wrt all completions of $\P$ if it satisfies JR on $\A$.
\end{proof}

\subsection{Stronger Representation Axioms}\label{sec:proportional-extended}

We now turn our attention to stronger representation axioms. These axioms require that cohesive sets of voters are represented proportionally to their size.
\citet{BrPe23a} introduced the axioms of PJR+ and EJR+. These axioms are based on PJR (proportional justified representation) and EJR (extended justified representation) that we discuss in \Cref{sec:pjr-ejr}. Let $\A$ be a (complete) approval profile and $W\subseteq C$ a committee of size $|W|=k$.
A group $G \subseteq V$ of voters is said to be \emph{$\ell$-large} for some $\ell \in \naturals$ if $|G| \geq \ell \cdot n / k$. In this section, we call $G$ \emph{cohesive} (\wrt $\A$ and $W$) if there exists a candidate $c \in C \setminus W$ such that $c \in \bigcap_{v \in G} A(v)$.
Note that, in contrast to JR, the requirement of cohesiveness depends on the committee $W$. 

Then, $W$ satisfies \emph{PJR+} \wrt $\A$ if for every integer $\ell \leq k$ and every voter  group $G \subseteq V$ that is $\ell$-large and cohesive, it holds that 
the committee $W$ contains at least $\ell$ candidates that are approved by some voter in $G$, i.e., 
\[ \left| W \cap \left( \bigcup_{v \in G} A(v) \right) \right| \geq \ell.\]
Further, $W$ satisfies \emph{EJR+} if for every $\ell \leq k$ and every voter group $G \subseteq V$ that is $\ell$-large and cohesive, there is at least one voter $v \in G$ for which $|W \cap A(v)| \geq \ell$.

Given a complete approval profile $\A$ and a committee $W$, it can be verified in polynomial time whether $W$ satisfies PJR+ and EJR+~\cite{BrPe23a}. Similarly to \Cref{sec:jr}, we study the problems of $\necpjrp$ and $\necejrp$, where given a committee $W$ and a partial approval profile $\P$, we wish to decide whether $W$ satisfies PJR+ or EJR+, respectively, in every completion of $\P$. We also consider $\pospjrp$ and $\posejrp$, where the goal is to decide whether $W$ satisfies PJR+ or EJR+, respectively, in at least one completion of $\P$.
As in \Cref{sec:jr}, the committee size $k$ is part of the input, and we do not assume it is fixed. 

As noted above, since the level of representation that these axioms demand scales with the cohesiveness and size of the voter sets, they do not exhibit the same ``monotonicity'' properties that JR does (see \Cref{lem:JRchangeprofile}). Nevertheless, we provide polynomial time algorithms for $\necpjrp$ and $\necejrp$ in the poset model, and for $\pospjrp$ and $\posejrp$ in the three-valued model.

\begin{theorem}\label{res:necpjrpejrp}
    In the poset approval model, $\necpjrp$ and $\necejrp$ are solvable in polynomial time.
\end{theorem}
\begin{proof}
    We start with $\necpjrp$. Consider a partial approval profile $\P$ and a committee $W$ of size $|W| = k$. For every candidate $c \in C \setminus W$, we define a completion $\A_{c}$ of $\P$ as follows. For every voter $v \in V$ with $c \in \middle(v)$, in $\A_{c}$ the voter approves $c$ and as few other candidates as possible. Voters $v$ with $c \not\in \middle(v)$ only approve $\top(v)$ in $\A_c$. 
    We claim that $W$ satisfies PJR+ for every completion of $\P$ if and only if $W$ satisfies PJR+ in $\A_c$ for all $c \in C \setminus W$.
    Since PJR+ can be verified efficiently~\cite{BrPe23a}, the latter can be checked in polynomial time.
    
    If $W$ does not satisfy PJR+ in $\A_c$ for some $c \in C \setminus W$, then we have a ``no'' instance by definition. 
    Now assume that there exists a completion $\A$ of $\P$ where $W$ does not satisfy PJR+. We will show that $W$ violates PJR+ in $\A_c$ for some $c \in C \setminus W$. Since $W$ violates PJR+ in $\A$, there exists an integer $\ell \leq k$, a voter group $G \subseteq V$, and a candidate $c \in C \setminus W$ such that $G$ is $\ell$-large, $c \in \bigcap_{v \in G} A(v)$, and $\left| W \cap \left( \bigcup_{v \in G} A(v) \right) \right| < \ell$.
    Consider the same voter group $G$ in the completion~$\A_c$. For every group member $v \in G$ we have $c \in A_c(v)$ and $A_c(v) \subseteq A(v)$, since $v$ approves $c$ in~$\A$. We can, therefore, deduce that $c \in \bigcap_{v \in G} A_c(v)$ and
    \[ \abs{W \cap \rpar{\bigcup_{v \in G} A_c(v)}} \leq \abs{W \cap \rpar{\bigcup_{v \in G} A(v)}} < \ell \,.\]
    This implies that $W$ does not satisfy PJR+ in $\A_c$, which completes the proof for $\necpjrp$.

    For $\necejrp$, the algorithm is the same, except that we check whether $W$ satisfies EJR+ in each $\A_c$ (instead of PJR+). The correctness follows from the same idea: If there exists a completion with an integer $\ell \leq k$ and a group of voters $G \subseteq V$ such that $G$ is $\ell$-large, cohesive, and $|W \cap A(v)| < \ell$ for every $v \in G$, then the same holds for $G$ in the completion $\A_c$.
\end{proof}

We now turn to the question whether a committee can possibly satisfy PJR+ or EJR+. For the 3VA model, we show that both problems are solvable in polynomial time.

\begin{theorem}\label{res:pospjrpejrp}
    In the 3VA model, $\pospjrp$ and $\posejrp$ are solvable in polynomial time.
\end{theorem}
\begin{proof}
    We start with $\pospjrp$, and use a similar algorithm to the proof of \Cref{thm:simplePoscomkApproval}. Given a partial profile $\P$ and a subset $W \subseteq C$ of size $k$, define a completion $\A$ where for every voter $v \in V$, $\middle_\A(v) = \middle(v) \cap W$. We show that $W$ satisfies PJR+ in at least one completion of $\P$ if and only if $W$ satisfies PJR+ in $\A$. 

    One direction is trivial. For the other direction,
    let $\A'$ be a completion where $W$ satisfies PJR+. We need to show that $W$ satisfies PJR+ \wrt $\A$.
    Let $G \subseteq V$ be an $\ell$-large group of voters that is cohesive in $\A$, \ie there exists $c \in C \setminus W$ such that $c \in \bigcap_{v \in G} A(v)$. By the construction of $\A$, we have $c \in \bigcap_{v \in G} \top(v)$, hence $G$ is also cohesive in $\A'$. The construction of $\A$ implies that $W \cap A'(v) \subseteq W \cap A(v)$ for all $v \in V$. 
    Together with the assumption that $W$ satisfies PJR+ in $\A'$, we can deduce that
    \[ \abs{W \cap \rpar{\bigcup_{v \in G} A(v)}} \geq \abs{W \cap \rpar{\bigcup_{v \in G} A'(v)}} \geq \ell \, \]
    which implies that $W$ satisfies PJR+ \wrt $\A$.

    For $\posejrp$ we use the same algorithm, except that we check whether $W$ satisfies EJR+ \wrt $\A$ (instead of PJR+). The correctness follows from the same argument that we used for $\pospjrp$.
\end{proof}

It is left open whether $\pospjrp$ and $\posejrp$ are solvable in polynomial time in the poset approval model, or even in the linear model.
The greedy approach in the proof of \Cref{res:pospjrpejrp} does not immediately transfer to the linear model, since approving additional candidates in $W$ can lead to sub-optimal completions (in terms of satisfying PJR+ and EJR+) if we need to approve candidates outside of $W$ as well. The following example illustrates this.

\begin{table}[t]
    \centering
    \scalebox{1}{
        \begin{tabular}{cccc}
            \toprule
            voter $v$ & $\top(v)$ & $\middle(v)$ & $\bottom(v)$ \\
            \midrule
            1 & $a, b$ & -- & $c$ \\
            2 & $a, b$ & -- & $c$ \\
            3 & -- & $a \rightarrow b$ & $c$ \\
            \bottomrule
        \end{tabular}
    }
    \caption{The incomplete approval profile from \Cref{ex:posPJRP_posEJRP}.}
    \label{tbl:example_posPJRP_necPJRP}
\end{table}

\begin{example}\label{ex:posPJRP_posEJRP}
    Consider the instance depicted in \Cref{tbl:example_posPJRP_necPJRP} using the model of linear incomplete approvals and assume $k=2$. There are 3 voters with incomplete approval preferences over 3 candidates, and the third voter is the only one with uncertainty. 
    The committee $W = \set{b, c}$ satisfies both PJR+ and EJR+ in the completion where $A(3) = \emptyset$. However, if $A(3) = \set{a, b}$, then $W$ satisfies neither PJR+ nor EJR+, since $G = \set{1, 2, 3}$ is a cohesive group that is $2$-large, and $A(v) \cap W = \set{b}$ for every $v \in G$.
\end{example}

\section{Discussion}
\label{sec:conclusion}

We studied computational aspects of approval-based committee voting under incomplete approval information. 
We adopted the poset approval as a general model of incompleteness, along with the 3VA and linear special cases, and we focused on committees of a fixed size. 
We established a quite broad picture of complexity for the problems of determining whether a given set of candidates is a possible or necessary committee for different classes of ABC rules.
Finally, we investigated the question of whether a given committee satisfies, possibly or necessarily, the representation axiom \emph{justified representation} or the stronger proportionality axioms \emph{PJR+} and \emph{EJR+}.
For the three axioms, we show that the questions of necessary representativeness can in fact be solved in polynomial time.
For the question of possible representativeness, we established an efficient algorithm for JR in the poset model, and for PJR+ and EJR+ in the 3VA model. The algorithms for possible and necessary representativeness do not assume the committee size to be fixed.

It seems promising to study the problem of possible and necessary committees also for other ABC rules such as Phragmén's rules \citep{BFJL24a} or the Method of Equal Shares \citep{PPS21b}.
Additional directions for future research include other models of incompleteness, and uncertainty in general, such as a model where voter attendance to the ballot is uncertain \citep{ImKi21a}.

\section*{Acknowledgments}

We thank Jérôme Lang for helpful discussions as well as Phokion Kolaitis and Julia Stoyanovich for insightful suggestions on modeling voter incompleteness. The work of Jonas Israel and Markus Brill was supported by the Deutsche Forschungsgemeinschaft (DFG) under Grant \mbox{BR 4744/2-1}. The work of Aviram Imber and Benny Kimelfeld was supported by the US-Israel Binational Science Foundation (BSF) under Grant 2017753.

\bibliographystyle{plainnat}
\bibliography{abb,algo,References} 

\newpage
\appendix

\section{Proportional Justified Representation and Extended Justified Representation}\label{sec:pjr-ejr}
In this section we study two additional axioms from the justified representation family: PJR (proportional justified representation) and EJR (extended justified representation) \citep{ABC+16a,SFF+17a}. These two axioms differ from PJR+ and EJR+ (discussed in \Cref{sec:proportional-extended}) by using a different notion of cohesiveness. Let $\A$ be a (complete) approval profile and $W\subseteq C$ a committee of size $|W|=k$. A group of voters $G \subseteq V$ is said to be \emph{$\ell$-cohesive} if $|\cap_{v \in G} A(v)| \geq \ell$. Then, $W$ satisfies PJR if for all $\ell \leq k$ and all $G \subseteq V$ that are $\ell$-large and $\ell$-cohesive, it holds that 
\[ \left| W \cap \left( \bigcup_{v \in G} A(v) \right) \right| \geq \ell.\]
$W$ satisfies EJR if for all $\ell \leq k$ and all $G \subseteq V$ that are $\ell$-large and $\ell$-cohesive, it holds that there is at least one $v \in G$ such that
\[ |W \cap A(v)| \geq \ell. \]
We again ask whether a committee possibly or necessarily satisfies these axioms given an incomplete profile. Before we consider the computational complexity of these questions, we add a small example.

\begin{example}\label{ex:posEJR_necEJR}
    Consider again the instance depicted in \Cref{tbl:example_posJR_necJR} with 6 voters and 6 candidates. 
    Again using a committee size of $k = 3$, we can see the following. Both $\set{b,d,e}$ and $\set{a,c,e}$ possibly satisfies PJR and EJR (\eg when voter 4 approves of $a$ and $d$ and all other voters $i \neq 4$ only approve candidates in $\top(i)$) but not necessarily (since voters 1 to 4 form a 2-cohesive group if all of them approve $a$ and $b$). But $\set{a,b,c}$ necessarily satisfies PJR and EJR since every voter is represented at least once and the only 2-cohesive group possible, voters 1 to 4, is represented twice.
\end{example}

In contrast to the axioms we study in \Cref{sec:representation}, already for a complete approval profile~$\A$, it is coNP-complete to decide whether a given committee provides PJR or EJR \wrt $\A$ \citep{ABC+16a,AEH+18a}. 
Therefore, in this section we make the same assumption that committee size $k$ is fixed. We begin by showing that the verification problems become tractable under this assumption. Our proof uses a similar idea as the proof by \citet{AEH+18a} for an analogous result for a fixed number of candidates (instead of a fixed committee size).

\begin{theorem}
\label{thm:testcomkPEJR}
Given a committee $W$ of size $k$ and a (complete) approval profile $\A$, it can be tested in polynomial time whether $W$ satisfies PJR or EJR, for all fixed $k \geq 1$.
\end{theorem}
\begin{proof}
    We prove the result for PJR first and then describe what changes when we consider EJR.
    Assume we are given a committee $W$ with $|W|=k$ and an approval profile $\A$. We now describe how to check if $W$ satisfies PJR with respect to $\A$ in time polynomial in the input size, where we treat $k$ as a constant. First, note that it is sufficient to show that we can check PJR for $\ell$-cohesive and $\ell$-large sets in polynomial time and then iterate over all $\ell \leq k$ to test PJR. We call the property of satisfying PJR for $\ell$-cohesive and $\ell$-large sets $\ell$-PJR. If a committee satisfies $\ell$-PJR for all $\ell \leq k$ it satisfies PJR. We thus fix some $\ell \leq k$.
    
    Now we iterate over all pairs of subsets $S \subseteq C$ of size $\ell$ and $W' \subseteq W$ of size at most $\ell-1$.  Note that the number of these subsets is polynomial since $k$ is fixed and $\ell \leq k$. Given $S$ and $W'$, define 
    \[V_{S, W'} = \set{v \in V : S \subseteq A(v) \land A(v) \cap W \subseteq W'}\]
    as the group of voters that approve all candidates of $S$, and among the candidates in $W$, only approve candidates in $W'$. We have $|\bigcap_{v \in V_{S, W'}} A(i)| \geq |S| = \ell$ and $\left| W \cap \left( \bigcup_{v \in V_{S, W'}} A(v) \right)\right| \leq |W'| < \ell$, \ie $V_{S, W'}$ is $\ell$-cohesive and not represented by $W$. Observe that $W$ violates $\ell$-PJR if and only if there is at least one pair $S, W'$ such that $V_{S, W'}$ is $\ell$-large.   
    Further, $W$ violates PJR if and only if there is at least one such $\ell$ for which $\ell$-PJR is violated.

    The algorithm for EJR is slightly simpler. For each $\ell \leq k$, we  iterate over each set $S \subseteq C$ of size $\ell$ (no need for $W'$) and define
    \[V_S = \set{v \in V : S \subseteq A(v) \land |A(v) \cap W| < \ell} \,.\]
    We can show that $W$ satisfies EJR if and only if there exists $\ell \leq k$ and $S$ for which $V_S$ is $\ell$-large, by a similar logic to the proof for PJR.
\end{proof}

\Cref{thm:testcomkPEJR} gives rise to computational questions regarding PJR and EJR under incompleteness: Given a committee $W$ and a partial approval profile $\P$ and assuming $k$ is fixed, can we decide in polynomial time whether $W$ satisfies PJR or EJR for any or all of the completions of $\P$? We answer the question affirmatively for the \emph{necessary} part, for every fixed $k \geq 1$. 
We call the corresponding problems $\necpjrpar{k}$ and $\necejrpar{k}$, respectively.

\begin{theorem}\label{res:necpjr}
    In the poset approval model, $\necpjrpar{k}$ is solvable in polynomial time, for all fixed $k \geq 1$.
\end{theorem}
\begin{proof}
We start with proving an auxiliary claim:
Assume $W$ fails to provide PJR on some completion $\A$ because there is an $\ell$-large, $\ell$-cohesive voter group $G \subseteq V$ with $|\bigcup_{v \in G} A(i) \cap W| < \ell$. Then $W$ also fails to provide PJR on any completion $\A'$ obtained from $\A$ where some voters additionally approve of candidates of $(C \setminus W) \cup \bigcup_{v \in G} A(v)$, \ie additionally approve of any candidates that are already approved by a voter in $G$ or are not in $W$. This is because the cohesiveness of sets in $\A'$ can only increase compared to $\A$ and $G$ still has less than $\ell$ candidates in $\bigcup_{v \in G} A'(v) \cap W$.

The algorithm for $\necpjrpar{k}$ works as follows. For each $W' \subseteq W$ we consider the completion $\A_{W'}$ where every voter approves of as many candidates in $C\setminus W'$ as possible without approving of any additional $c \in W'$. This can be done efficiently by iterating for each voter $v$ over the candidates in $\middle(v)$ and checking for each candidate $c \in \middle(v)$ whether $c$ can be approved without needing to approve a candidate from $W'$. If so, the candidate approves $c$.
We then check whether $W$ satisfies PJR in $\A_{W'}$ as is described in \Cref{thm:testcomkPEJR}. The algorithms terminates and outputs \emph{no} as soon as a PJR violation is found and terminates and outputs \emph{yes} if all $W' \subseteq W$ passed the check. Since $k$ is fixed, we can iterate over all subsets of $W$ in polynomial time.

For the correctness of the algorithm first note that for $W$ to violate PJR in a completion $\A$ there has to be $W' \subset W$, $G \subseteq V$ and integer $\ell \leq k$ such that $G$ is $\ell$-large, $\ell$-cohesive, and $W' = W \setminus \bigcup_{v \in G} A(v)$ satisfies $|W'| \geq k-\ell$.
(This is equivalent to $|W \cap \bigcup_{v \in G}| < \ell$.) In that case we say $W'$ is a \emph{witness for PJR failure (\wrt $G$ and $\ell$)}.
If no $W' \subseteq W$ is such a witness (for any $G$ and $\ell$) then $W$ satisfies PJR.
Now fix some $W' \subset W$ and let $\A_{W'}$ be the completion described in the algorithm above. We will now show that if there is a completion $\A$ for which $W'$ is a witness \wrt some $G$ and $\ell$ then it is also a witness \wrt the same $G$ and $\ell$ on $\A_{W'}$.

So assume that there is some $\A$ where $W$ fails PJR and $W'$ is a witness \wrt some $G$ and $\ell$. Let $v \in G$. By the definition of witness above, $\middle_\A(v) \cap W' = \emptyset$. In $\A_{W'}$, $v$ approves as many candidates outside $W'$ as possible without approving any additional $c \in W'$, hence $A(v) \subseteq A_{W'}(v)$ and $A_{W'}(v) \setminus A(v) \subseteq C \setminus W'$.
We get that when we change $A(v)$ to $A_{W'}(v)$, we additionally approve candidates that are already approved by a voter in $G$ or a not in $W$.
Further note that for all $v \notin G$ we can change the approvals in any way without changing the fact that $W'$ is a witness of PJR failure \wrt $G$ and $\ell$.
This together with the claim we proved in the beginning implies the desired result: It is enough to check one specific completion for each $W' \subseteq W$ when searching for possible PJR failures of $W$.
\end{proof}

\begin{theorem}\label{res:necejr}
    In the poset approval model, $\necejrpar{k}$ is solvable in polynomial time, for all fixed $k \geq 1$.
\end{theorem}
\begin{proof}
The algorithm to solve $\necejrpar{k}$ in polynomial time works as follows. We iterate over all subsets of candidates $C' \subseteq C$ with $|C'| \leq k$. For $C'$ we define a completion $\A_{C'}$ where each voter approves of as many candidates in $C'$ as possible (and whatever candidates are necessary to approve the candidates in $C'$) but not more.
We then check whether $W$ satisfies EJR in the completion $\A_{C'}$ as described in \Cref{thm:testcomkPEJR}.
The algorithm terminates and outputs \emph{no} as soon as an EJR violation is found and terminates and outputs \emph{yes} if $C' \subseteq C$ with $|C'| \leq k$ passed the check. As in the algorithm of \Cref{res:necpjr}, the iteration exploits the fact that $k$ is fixed.

For the correctness of the algorithm assume there is a completion $\A$ where $W$ fails EJR, \ie there is $\ell \leq k$ and $G \subseteq V$ that is $\ell$-large, $\ell$-cohesive with $|W \cap A(v)| < \ell$ for all $v \in G$. Let $C'$ be some subset of $\bigcap_{v \in G} A(v)$ of size $\ell$. (It exists since $G$ is $\ell$-cohesive.) We show that $W$ then also fails EJR in the $\A_{C'}$. To that end, note that for every $v \in G$ we have $C' \subseteq A_{C'} \subseteq A(v)$ and $G$ is $\ell$-cohesive. Therefore $G$ is $\ell$-cohesive in $\A_{C'}$, and the representation of $G$ by $W$ cannot increase. Thus $W$ fails EJR in $\A_{C'}$.
This proves that checking $\A_{C'}$ for every $C' \subseteq C$ of size at most $k$ is enough when searching for EJR violations of $W$.
\end{proof}

The complexity of $\pospjrpar{k}$ and $\posejrpar{k}$ remains open in all three models, including the 3VA model in which $\pospjrp$ and $\posejrp$ are solvable in polynomial time (see \Cref{res:pospjrpejrp}). The greedy approach of \Cref{res:pospjrpejrp} does not work for PJR and EJR, because approving additional candidates in $W$ is not necessarily beneficial to the satisfaction of PJR and EJR (in contrast to PJR+ and EJR+, where approving additional candidates in $W$ is problematic only if it requires us to approve candidates outside of $W$). The following example illustrates a case of this type.

\begin{table}[t]
    \centering
    \scalebox{1}{
        \begin{tabular}{cccc}
            \toprule
            voter $v$ & $\top(v)$ & $\middle(v)$ & $\bottom(v)$ \\
            \midrule
            1 & $a, b$ & -- & $c$ \\
            2 & $a, b$ & -- & $c$ \\
            3 & $a$ & $b$ & $c$ \\
            \bottomrule
        \end{tabular}
    }
    \caption{The incomplete approval profile from \Cref{ex:posEJR_posEJR}.}
    \label{tbl:example_posPJR_necPJR}
\end{table}

\begin{example}\label{ex:posEJR_posEJR}
    Consider the instance depicted in \Cref{tbl:example_posPJR_necPJR}. We can use either of the three models, since $\middle(v) = \emptyset$ or $|\middle(v)| = 1$ for every voter. Here three voters have incomplete approval preferences over three candidates, and there are two possible completions: $\A$ where $A(3) = \set{a}$ and $\A'$ where $A'(3) = \set{a, b}$.
    For the committee $W = \set{b, c}$, we can see the following. $W$ satisfies both PJR and EJR in $\A$, but it does not satisfy PJR and EJR in $\A'$: in $\A'$ the voter group $G = \set{1, 2, 3}$ is $2$-cohesive, $2$-large, and $A'(v) \cap W = \set{b}$ for every $v \in G$.
    Note that $\A'$ is the completion that the algorithm of \Cref{res:pospjrpejrp} would construct for $W$.
\end{example}

\end{document}